\documentclass[12pt]{article}

\usepackage{amsmath}
\usepackage{amssymb}
\usepackage{amsthm}
\usepackage{hyperref}
\usepackage{enumerate}
\usepackage{graphicx}
\usepackage{color}
\usepackage{mathtools}
\usepackage{float}

\newtheorem{thm}{Theorem}[section]

\newtheorem{lem}[thm]{Lemma}

\newtheorem{assumption}[thm]{Assumption}

\newtheorem{definition}[thm]{Definition}

\newtheorem{example}[thm]{Example}

\newtheorem{remark}[thm]{Remark}
\newenvironment{rem}{\begin{remark}\rm}{\end{remark}}
\newtheorem{algorithm}[thm]{Algorithm}

\newcommand{\cob}{\color{black}}

\title{A Deterministic Algorithm for the Capacity of Finite-State Channels~\footnote{A preliminary version~\cite{WuHanMarcus} of this work has been presented in IEEE ISIT 2019.}~\thanks{This research is partly supported by a grant from the Research Grants Council of the Hong Kong Special Administrative Region, China (Project No. 17301017) and a grant by the National Natural Science Foundation of China (Project No. 61871343).
VA acknowledges support from NSF grants CNS--1527846, CCF--1618145, CCF-1901004, the NSF Science \& Technology Center grant CCF--0939370 (Science of Information), and the William and Flora Hewlett Foundation
supported Center for Long Term Cybersecurity at Berkeley.
}}

\author{Chengyu Wu\textsuperscript{1}, Guangyue Han\textsuperscript{2}, Venkat Anantharam\textsuperscript{3} and Brian Marcus\textsuperscript{4}\\[2ex]
\textsuperscript{1}The University of Hong Kong, \textit{chengyuw@connect.hku.hk}\\
\textsuperscript{2}The University of Hong Kong, \textit{ghan@hku.hk}\\
\textsuperscript{3}University of California, Berkeley, \textit{ananth@berkeley.edu}\\
\textsuperscript{4}The University of British Columbia, \textit{marcus@math.ubc.ca}}

\date{{\normalsize \today}}
\begin{document}\maketitle\thispagestyle{empty}

\begin{abstract}
We propose two modified versions of the classical gradient ascent method to compute the capacity of finite-state channels with Markovian inputs. For the case that the channel mutual information is strongly concave
in a parameter taking values in a compact convex subset of
some Euclidean space,
\cob
our first algorithm proves to achieve polynomial accuracy
in polynomial time and, moreover,
\cob
for some special families of finite-state
channels
\cob
our algorithm can achieve exponential accuracy in polynomial time
under some technical conditions.
\cob
For the case that the channel mutual information may not be strongly concave, our second algorithm proves to be at least locally convergent.
\end{abstract}

\section{Introduction}

As opposed to a discrete memoryless channel, which features a single state and thereby can be characterized by input and output random variables only, the characterization of a finite-state channel has to resort to additional state random variables.
Encompassing discrete memoryless channels as special cases, finite-state channels have long been used in a wide range of communication scenarios where the current behavior of the channel may be affected by its past. Among many others, conventional examples of such channels include inter-symbol interference channels~\cite{fo72}, partial response channels~\cite{pr00,th87} and Gilbert-Elliott channels~\cite{mu89}.

While it is well-known that the Blahut-Arimoto algorithm~\cite{ar72, bl72} can be used to efficiently compute the capacity of a discrete memoryless channel, the computation of the capacity of a general finite-state channel has long been a notoriously difficult problem, which has been open for decades. The difficulty of this problem may be justified by the widely held (yet not proven) belief that the capacity of a finite-state channel may not be achieved by any finite-order Markovian input, and an increase of the memory of the input may lead to an increase of the channel capacity.

We are mainly concerned with finite-state channels with Markov processes of a fixed order as their inputs. Possibly an unavoidable compromise we have to make in exchange for progress in computing the capacity, the extra fixed-order assumption imposed on the input processes is also necessary for the situation
where the channel input
\cob
has to satisfy certain constraints, notably finite-type constraints~\cite{lm95} that are commonly used in magnetic and optical recording. On the other hand, the focus on Markovian inputs can also be justified by the known fact that the Shannon capacity of an indecomposable finite-state channel~\cite{ga68} can be approximated by the Markov capacity with increasing orders (see Theorem $2.1$ of \cite{LiHanSiegel}). Recently, there has been some
progress in computing
\cob
the capacity of finite-state channels with such input constraints. Below we only list the most relevant work in the literature, and we refer the reader to~\cite{han15} for a comprehensive
list of references.
\cob
In~\cite{ka01}, the Blahut-Arimoto algorithm was reformulated into a stochastic expectation-maximization procedure and a similar algorithm for computing the lower bound of the capacity of finite-state channels was proposed, which led to a generalized Blahut-Arimoto algorithm \cite{pa04} that proves to compute the capacity under some concavity assumptions. More recently, inspired by ideas in stochastic approximation, a randomized algorithm was proposed~\cite{han15} to compute the capacity under weaker concavity assumptions, which can be verified to hold true for families of practical channels~\cite{hm09b,LiHan}. Both of the above-mentioned algorithms, however,
are of a randomized nature
\cob
(a feasible implementation of the generalized Blahut-Arimoto algorithm will necessitate a randomization procedure).
By comparison, among many other advantages, our algorithms, which are deterministic in nature, can be used to derive accurate estimates on the channel capacity, as evidenced by the tight bounds in Section~\ref{applications}.

In this paper, we first deal with the case that the mutual information of the finite-state channel is strongly concave
in a parameter taking values in a compact convex subset of some Euclidean space,
\cob
for which we  propose our first algorithm that proves to converge to the channel capacity
exponentially fast.
\cob
This algorithm largely follows the spirit of the classical gradient ascent method. However, unlike the classical case, the lack of an explicit expression
for our target function
\cob
and the boundedness of the variable domain (without an explicit description of the boundary) pose additional challenges. To overcome the first issue, a convergent sequence of approximating functions (to the original target function)
is used
\cob
instead in our treatment; meanwhile, an additional check condition is also added to ensure
that the iterates stay inside
\cob
the given variable domain. A careful convergence analysis has been carried out to deal
with the difficulties
\cob
caused by such modifications. This algorithm is efficient in the sense that, for a general finite-state channel (satisfying the above-mentioned concavity condition
and some additional technical conditions),
\cob
it achieves polynomial accuracy in polynomial time (see Theorem~\ref{general-channel}), and for some special families of finite-state
channels
\cob
it achieves exponential accuracy in polynomial time (see Section~\ref{applications}).

It is well known that the mutual information of a finite-state channel may not be concave
under the natural parametrization in several examples;
\cob
see, e.g.,~\cite{hm09b, LiHan}. Another modification of the classical gradient ascent method is proposed to handle this challenging scenario. Similar to our first algorithm, our second one replaces the original target function with a sequence of approximating functions, which unfortunately renders conventional methods such as the Frank-Wolfe method (see, e.g., \cite{be99}) or methods using the \L ojasiewicz inequality (see, e.g.,~\cite{ab06}) inapplicable. To address this issue, among other subtle modifications, we impose an extra check in the algorithm to slow down the pace ``a bit'' to avoid an immature convergence to a non-stationary point but ``not too much'' to ensure the local convergence.

As variants of the classical gradient ascent method, our algorithms can be applied to any sequence of convergent functions,
so they can be
\cob
of particular interest in information theory since many information-theoretic quantities are defined as the limit of their finite-block versions. On the other hand though, we would like to add that our algorithms are actually stated in much more general settings and may have potential applications in optimization scenarios where the target functions are difficult to compute but amenable to approximations.

The remainder of this paper is organized as follows. In Section \ref{channel-model}, we describe our channel model in great detail. Then, we present our first algorithm (Algorithm~\ref{algo1}) in Section~\ref{algorithm} and analyze its convergence behavior in Section~\ref{convergence} under some strong concavity assumptions. Applications of this algorithm for computing the capacity of finite-state channels under concavity assumptions will be discussed in Section~\ref{applications}. In particular, in this section, we show that the estimation of the channel capacity can be improved by increasing the Markov order of the input process
in some examples.
\cob
In Section~\ref{without concavity}, our second algorithm (Algorithm~\ref{algo3}) is presented, which proves to be at least locally convergent. Finally, in Section \ref{gil-ell}, our second algorithm
is applied
\cob
to a Gilbert-Elliott channel where the concavity of the channel mutual information rate
in the natural parametrization
\cob
is not known, and yet
fast convergence behavior
\cob
is observed.

In the remainder of this paper, the base of the logarithm is assumed to be $e$.
\cob

\section{Channel Model and Problem Formulation} \label{channel-model}

In this section, we introduce the channel model considered in this paper, which is essentially the same as that in~\cite{han15, pa04}.

As mentioned before, we are concerned with a discrete-time finite-state channel with a Markovian channel input. Let $X=\{X_n: n = 1, 2, \dots\}$ denote the channel input process, which is often assumed to be a first-order stationary Markov chain~\footnote{The assumption that $X$ is a first-order Markov chain is
for notational convenience
\cob
only: through a usual ``reblocking'' technique, the higher-order Markov case can be boiled down to the first-order case.} over a finite alphabet $\mathcal{X}$, and let $Y=\{Y_n: n =1, 2, \dots\}$ and $S=\{S_n: n =0, 1, \dots\}$ denote the channel output and state processes over finite alphabets $\mathcal{Y}$ and $\mathcal{S}$, respectively.

Let $\Pi$ be the set of all the stochastic matrices of dimension $|\mathcal{X}| \times |\mathcal{X}|$. For any finite set $F \subseteq \mathcal{X}^2$ and any $\delta > 0$, define
$$
\Pi_{F, \delta}\triangleq\{A \in \Pi: A_{i j}=0,~~\mbox{for}~(i, j) \in F ~~\mbox{and}~ A_{ij}\geq \delta ~~\mbox{otherwise}\}.
$$
It can be easily verified that if one of the matrices from $\Pi_{F, \delta}$ is primitive, then all matrices from $\Pi_{F, \delta}$ will be primitive, in which case,
as elaborated on in~\cite{han15},
\cob
$F$ gives rise to a so-called {mixing} finite-type constraint. Such a constraint has been widely used in data storage and magnetic recoding~\cite{mrs98}, the best known example being the so-called $(d,k)$-run length limited (RLL) constraint over the alphabet $\{0, 1\}$, which forbids any sequence with fewer than $d$ or more than $k$ consecutive zeros in between two successive $1$'s.

The following conditions will be imposed on the finite-state channel described above:
\begin{enumerate}
\item [(\ref{channel-model}.$a$)] There exist $F \subseteq \mathcal{X}^2$ and $\delta > 0$ such that the transition probability matrix of $X$ belongs to $\Pi_{F, \delta}$, each element of which is a primitive matrix.
\item[(\ref{channel-model}.$b$)] $(X, S)$ is a first-order stationary Markov chain whose transition probabilities satisfy
$$
p(x_n, s_n|x_{n-1}, s_{n-1})=p(x_n|x_{n-1}) p(s_n|x_n, s_{n-1}), \quad n = 1, 2, \dots,
$$
where $p(s_n|x_n, s_{n-1}) > 0$ for any $s_{n-1}, s_{n}, x_n$.
\item[(\ref{channel-model}.$c$)] The channel is stationary and characterized by
$$
p(y_n|y_1^{n-1}, x_1^n, s_1^{n-1})=p(y_n|x_n, s_{n-1}) > 0, \quad n = 1, 2, \dots,
$$
that is, conditioned on the pair
$(x_n, s_{n-1})$,
the output $Y_n$ is statistically independent of all
inputs, outputs and states
prior to $X_n, Y_n$ and $S_{n-1}$, respectively.
\end{enumerate}
As elaborated on in
Remark 4.1 of \cite{han15}, a finite-state channel specified as above is indecomposable. Therefore, assuming that the input $X$ (or, more precisely, the transition probability matrix of $X$) is
analytically
parameterized
by a finite-dimensional parameter $\theta$ in a compact convex subset $\Theta$ of some Euclidean space (such a parameterization exists thanks to the stationarity of $X$), we can express the capacity of the above channel as
\begin{equation} \label{C-I}
C=\max_{\theta\in \Theta} I(X(\theta); Y(\theta))= \max_{\theta\in \Theta} \lim_{k \to \infty} I_k(X(\theta); Y(\theta)),
\end{equation}
where
\begin{equation} \label{I-n}
I_k(X(\theta); Y(\theta)) \triangleq \frac{H(X_1^k(\theta))+H(Y_1^k(\theta))-H(X_1^k(\theta), Y_1^k(\theta))}{k}.
\end{equation}
Moreover, it has also been shown in \cite{han15} that $I_k(X(\theta); Y(\theta))$ (\textit{resp.}, its derivatives) converges to $I(X(\theta); Y(\theta))$ 
(\textit{resp.}, the corresponding derivatives)
\cob
exponentially fast in $k$
\cob
under Assumptions (\thesection.$a$), (\thesection.$b$) and (\thesection.$c$). Hence, although the value of the target function $I(X(\theta);Y(\theta))$ cannot be exactly computed, it can be approximated
by the function $I_k(X(\theta); Y(\theta))$, which has an explicit expression, within an error exponentially decreasing in $k$.
\cob

Instead of merely solving (\ref{C-I}), we will deal with the following slightly more general problem
\begin{align} \label{general_opti}
&\max f(\theta)=\lim_{k\rightarrow \infty} f_k(\theta)  \notag \\
&\mbox{subject to} \quad \theta\in \Theta,
\end{align}
\cob
under the following assumptions:
\begin{enumerate}
\item $\Theta$ is a compact convex subset of $\mathbb{R}^d$
for some $d\in \mathbb{N}$ with nonempty interior $\Theta^\circ$ and boundary $\partial \Theta$;
\cob
\item
$f(\theta)$ and all $f_k(\theta)$, $k \geq 0$, are continuous on $\Theta$
\cob
and twice continuously differentiable in $\Theta^\circ$;
\item there exist $M_0>0$, $N > 0$ and $0<\rho<1$ such that for all $k\geq 1$, $\theta\in \Theta^\circ$ and $\ell=0,1,2$, it holds true that $||f_0^{(\ell)}(\theta)||_2\leq M_0$ and
\begin{equation} \label{C-rho}
||f_k^{(\ell)}(\theta)-f_{k-1}^{(\ell)}(\theta)||_2\leq N \rho^k, \quad \quad  ||f_k^{(\ell)} (\theta)-f^{(\ell)} (\theta)||_2\leq N \rho^k,
\end{equation}
where the superscript $^{(\ell)}$ denotes the $\ell$-th order derivative and $||\cdot ||_2$ denotes the Frobenius norm of a vector/matrix.
\end{enumerate}
Obviously, if we set $f_k(\theta)=I_k(X(\theta); Y(\theta))$ and assume that $X(\theta)$
analytically
parameterized
by some $\theta\in \Theta$, then (\ref{general_opti}) boils down to (\ref{C-I}).

When the target function $f(\theta)$ has an explicit expression and $\Theta$ is characterized by finitely many twice continuously differentiable constraints,
the optimization problem (3) can be effectively dealt with via, for example, the classical gradient ascent method \cite{boyd04} or the Frank-Wolfe method \cite{be99} or their numerous variants.
However, feasible implementations and executions of these algorithms usually hinge on explicit descriptions of $\Theta$ and $\nabla f$, both of which can be rather intricate in our setting.


Before moving to the next two sections to present our algorithms, we make some observations about the sequence $\{f_k(\theta)\}_{k=0}^\infty$. It immediately follows from the uniform boundedness of $||f_0^{(\ell)}(\theta)||_2$ and the inequality (\ref{C-rho}) that there exists $M>0$ such that for all $k\geq 0$, $\ell=0,1,2$ and $\theta\in \Theta^\circ$,
\begin{equation} \label{m-M}
||f_k^{(\ell)}(\theta)||_2\leq M.
\end{equation}
In particular, for any $\theta\in \Theta^\circ$, when $\ell=2$, $f_k^{(\ell)}(\theta)=\nabla^2 f_k(\theta)$ is a symmetric matrix
whose spectral norm is given by
$$
|||\nabla^2 f_k(\theta)|||_2\triangleq\sup_{\mathbf{x}\neq 0} \frac{||\nabla^2 f_k(\theta)\cdot \mathbf{x}||_2}{||\mathbf{x}||_2}=|\lambda_1(\theta)|,
$$
where $\lambda_1$ denotes the largest (in modulus) eigenvalue of $\nabla^2 f_k(\theta)$.
Hence, the inequality (\ref{m-M}) and the easily verifiable fact that $|||\nabla^2 f_k(\theta)|||_2\leq ||\nabla^2 f_k(\theta)||_2$ imply
\begin{equation} \label{succeq_M}
-M \mathbb{I}_d\preceq \nabla^2 f_k(\theta)\preceq M \mathbb{I}_d
\end{equation}
\cob
for any $k$ and any $\theta\in \Theta^\circ$, where $\mathbb{I}_d$ denotes the $d\times d$ identity matrix, and for two matrices $A, B$ of the same dimension, by $A\preceq B$, we mean that $B-A$ is a positive semidefinite matrix. The existence of the constant $M$ in (\ref{succeq_M}) will be crucial for implementing our algorithms.

\section{The First Algorithm: with Concavity} \label{algorithm}

Throughout this section, we assume that $f(\theta)$ is strongly concave,
i.e.,
\cob
there exists $m>0$
such that
\cob
for all $\theta\in \Theta^\circ,$
\begin{equation} \label{little-m}
\nabla^2 f(\theta) \preceq -m \mathbb{I}_d,
\end{equation}
and moreover
\cob
\begin{equation} \label{unique-maximum}
f \mbox{ achieves its unique maximum in } \Theta^\circ.
\end{equation}
We will present our first algorithm to solve the optimization problem  (\ref{general_opti}). As mentioned before, the algorithm is in fact a modified version of the classical gradient ascent algorithm, whereas its convergence analysis
is more intricate than
\cob
the classical one. To overcome the issue that the target function $f(\theta)$ may not have an explicit
expression
\cob
we capitalize on the fact that it can be well approximated by $\{f_k(\theta)\}_{k=0}^\infty$, which will be used instead to compute the estimates in each iteration.

Before presenting our algorithm, we need the following lemma, which, as evidenced later,
is important in
\cob
initializing and analyzing our first algorithm.
\begin{lem} \label{initial}
There exists a non-negative integer $k_0$ such that
\begin{enumerate}
\item[(a)] $\displaystyle\frac{(N+M)M \rho^{k_0+1}+2N \rho^{k_0+1}}{1-\rho}\leq \frac{\delta}{8}$ and $N \rho^{k_0}\leq \displaystyle\frac{\delta}{8}$, where $\delta \triangleq\max\limits_{\theta\in \Theta} f(\theta) - \max\limits_{\theta\in \partial \Theta} f(\theta) >0$.
\item[(b)] For any $k\geq k_0$, $f_{k}(\theta)$ is strongly concave and has a unique maximum in $\Theta^\circ$; and moreover, we have
\begin{align} \label{contr_perturbation_of_theta_k}
\sup_{k\geq k_0}||\theta_{k}^*-\theta^*||_2+\frac{d^{1/2}\rho^{k_0}}{1-\rho} < dist (\theta^*, \partial \Theta),
\end{align}
where $\theta^*$ denotes the unique maximum point of $f$ and $\theta^*_k$ denotes the unique maximum point of $f_k$.
\item[(c)] There exists $y_0\in \mathbb{R}$ such that
for all $k \ge k_0$,
\cob
$$\emptyset \subsetneq B_{k}\subseteq C_{k}\subseteq \Theta^\circ \quad \mbox{and} \quad dist (C_{k},\partial \Theta)>0,$$
\cob
where 
$$B_{k}\triangleq \{x\in \Theta: f_{k}(x)\geq y_0\} \quad \mbox{and} \quad C_{k}\triangleq \left\{x\in \Theta: f_{k}(x)\geq y_0-\frac{\delta}{8} \right\}.$$
\cob
\end{enumerate}
\end{lem}

\begin{proof}
Since $(a)$ trivially holds for sufficiently large $k_0$, we will omit its proof and proceed to prove $(b)$. Towards this end, note that according to (\ref{C-rho}) and (\ref{little-m}), it holds true that for sufficiently large $k$, each $f_k$ is strongly concave. Noting that $f(\theta^*)-\max_{\theta\in \partial \Theta} f(\theta)= \delta$, we deduce from $(a)$ and (\ref{C-rho}) that for $k$ large enough,
\begin{align} \label{max-boundary_of_k}
\max_{\theta\in \Theta}f_k(\theta)-\max_{\theta\in \partial \Theta} f_k(\theta)\geq f_k(\theta^*)-\max_{\theta\in \partial\Theta} f(\theta) -\frac{\delta}{8}\geq f(\theta^*)-\max_{\theta\in \partial \Theta} f(\theta) -\frac{\delta}{4}=\frac{3\delta}{4}>0.
\end{align}
\cob
Hence, for $k$ sufficiently large, $f_k$ achieves its unique maximum at $\theta_k^* \in \Theta^\circ$.

We now prove that 
$\theta_k^* \rightarrow \theta^*$ as $k \rightarrow \infty$. To see this, observe that (\ref{C-rho}) implies the uniform convergence of $f_k$ to $f$, i.e., for any $\varepsilon>0$, there exists $K>0$ such that for any $k>K$
and any $\theta\in \Theta$,
\cob
$f(\theta)-\varepsilon \leq f_k(\theta)\leq f(\theta) +\varepsilon.$
In particular, for $k> K$, we have
$$
f(\theta^*)-\varepsilon \leq f_k(\theta^*) \leq f_k(\theta_k^*)\leq f(\theta_k^*)+\varepsilon\leq f(\theta^*)+\varepsilon,
$$
which further
implies that $f_k(\theta_k^*) \rightarrow f(\theta^*)$
as $k \to \infty$.
\cob
It then follows from the triangle inequality that
\begin{equation} \label{f(theta_k)_tof(theta)}
f(\theta_k^*)\rightarrow f(\theta^*),~~\mbox{ as $k \to \infty$.}
\end{equation}
\cob
Now, by the Taylor series expansion, there exists some $\tilde{\theta}\in \Theta^\circ$ such that
\begin{equation} \label{taylor}
f(\theta_k^*)-f(\theta^*)=\nabla f(\theta^*)^T(\theta_k^*-\theta^*)+(\theta_k^*-\theta^*)^T\nabla^2 f(\tilde{\theta}) (\theta_k^*-\theta^*).
\end{equation}
Since $\nabla f(\theta^*)=0$ and $\nabla^2 f(\tilde{\theta}) \preceq -m \mathbb{I}_d$ according to (\ref{little-m}), it follows from (\ref{f(theta_k)_tof(theta)}) and (\ref{taylor}) that
$\theta_k^*\rightarrow \theta^*$ as $k \rightarrow \infty$,
\cob
as desired.

It then immediately follows that $||\theta_{k}^*-\theta^*||_2+d^{1/2}\rho^{k}/(1-\rho)\rightarrow 0$ as $k\rightarrow \infty$.
Observing that
\cob
$dist(\theta^*, \partial\Theta)>0$ (since $\theta^*\in \Theta^\circ$), we infer that (\ref{contr_perturbation_of_theta_k}) holds for sufficiently large $k$. Hence, $(b)$ will be satisfied as long as $k_0$ is sufficiently large.

We now show that $(c)$ also holds for sufficiently large $k_0$.
From the definition of $\delta$, there exists $y_0$ such that $\max_{\theta\in \partial \Theta} f(\theta)+\frac{\delta}{4}<y_0<\max_{\theta\in \Theta} f(\theta)-\frac{\delta}{4}$.
From (\ref{C-rho}), using the same logic as that used to derive
(\ref{max-boundary_of_k}), we infer that for sufficiently large $k$,
\begin{equation}        \label{interior}
\max_{\theta\in \partial\Theta} f_k(\theta)<y_0-\frac{\delta}{8}<y_0<\max_{\theta\in \Theta} f_k(\theta).
\end{equation}
\cob
According to $(b)$ and the fact that $\theta_k^*\in \Theta^\circ$,
which follows from (\ref{interior}),
\cob
we deduce that $\emptyset \subsetneq B_k\subseteq C_k\subseteq \Theta^o$ and $dist(C_k, \partial \Theta)>0$
with
$$
C_k\triangleq \left\{x: f_k(x)\geq y_0-\frac{\delta}{8}\right\} \quad \mbox{and} \quad B_k\triangleq \{x: f_k(x)\geq y_0\}.
$$
Therefore, $(c)$ is valid as long as $k_0$ is sufficiently large. Finally, choosing a larger $k_0$ if necessary, we conclude that there exists $k_0$ such that $(a)$, $(b)$ and $(c)$ are all satisfied.
\end{proof}
\begin{rem}
We remark that, for any $k \geq k_0$,
\cob
each $B_k$ specified as above has a non-empty interior,
which is due to the strict inequality (\ref{interior}) and the continuity of $f_k$.
\cob
\end{rem}

We are now ready to present our first algorithm, which modifies the classical gradient ascent method in the following manner: Instead of using
$\nabla f$
\cob
to find a feasible direction, we use $\nabla f_k$ as the ascending direction in the $k$-th iteration
and then pose
\cob
additional check conditions for a careful choice of the step size. Note that such modifications make the convergence analysis
more difficult
\cob
compared to the classical case,
as elaborated on in
\cob
the next subsection.

\begin{algorithm} \label{algo1}  (The first modified gradient ascent algorithm)

\textbf{Step $0$.} Choose $k_0$ such that Lemma \ref{initial} $(a)$-$(c)$ hold. Set $k=0$, $g_0=f_{k_0}$ and choose $\alpha \in (0, 0.5)$, $\beta \in (0, 1)$ and $\theta_{0} \in \Theta^\circ$ such that $\theta_{0}\in B_{k_0}$ and $\nabla g_{0}(\theta_{0}) \neq 0$.

\textbf{Step $1$.} Increase $k$ by $1$, and set $t=1$, $g_k=f_{k_0+k}$.

\textbf{Step $2$.} If $\nabla g_{k-1} (\theta_{k-1})=0$, set
$$
\tau=\theta_{k-1}+t\nabla g_{k-1} (\theta_{k-1}+\rho^{k+k_0} \mathbf{1}),
$$
where $\mathbf{1}$ denotes the all-one vector in $\mathbb{R}^d$;
\cob
otherwise, set
$$
\tau = \theta_{k-1} + t \nabla g_{k-1}(\theta_{k-1}).
$$
If $\tau \not \in \Theta$ or
$$
g_k(\tau) < g_k(\theta_{k-1}) + \alpha t ||\nabla g_{k-1}(\theta_{k-1})||_2^2 - (N+M) M t \rho^{k+k_0},
$$
set $t=\beta t$ and go to Step $2$, otherwise set $\theta_k=\tau$ and go to Step $1$.
\end{algorithm}

\begin{rem}
It is obvious from the definition of $g_k$ that as $k$ tends to infinity, $g_k$ (\textit{resp.}, its first and second order derivatives) converges to $f$ (\textit{resp.}, its first and second order derivatives) exponentially with the same constant as in (\ref{C-rho}).
\end{rem}

\begin{rem}
The existence of $\theta_{0}$ in Step $0$ can be justified by Lemma \ref{initial} $(c)$.
\end{rem}

\begin{rem} \label{perturbation_not_in_domain}
We point out that in Step 2 of Algorithm \ref{algo1}, for any $k\geq 1$,
the point $\theta_{k-1}+\rho^{k+k_0} \mathbf{1}$
\cob
will always lie in $\Theta^\circ$.
To see this,
note that, according to Algorithm \ref{algo1},
\cob
$\tau=\theta_{k-1}+t \nabla g_{k-1}(\theta_{k-1}+\rho^{k+k_0} \mathbf{1})$
\cob
only if $\theta_{k-1}$ is the maximum point of $g_{k-1} = f_{k+k_0-1}$, i.e., $\theta_{k-1}=\theta_{k+k_0-1}^*$. However, by Lemma \ref{initial} $(b)$, these points always satisfy (\ref{contr_perturbation_of_theta_k}), which immediately implies that
$\theta_{k-1}+\rho^{k+k_0} \mathbf{1} \in \Theta^\circ$
\cob
for all $k\geq 1$.

\end{rem}

\begin{rem}
For technical reasons that will be made clear in the next section, $\alpha$ is chosen within $(0,0.5)$ to ensure the convergence of the algorithm. In Step $2$ of algorithm \ref{algo1}, the case that $\nabla g_{k-1} (\theta_{k-1})=0$ is singled out for special treatment to prevent the algorithm from getting trapped at the maximum point of $f_{k-1}$ for a fixed $k$, which may be still far away from the maximum point of $f$.
\end{rem}

\subsection{Convergence Analysis} \label{convergence}


As mentioned earlier, compared to the classical gradient ascent method, Algorithm \ref{algo1} poses additional challenges for convergence analysis. The main difficulties come from the two check conditions in Step 2: the ``perturbed'' Armijo condition (see, e.g., Chapter 2 of \cite{be99} for more details)
$$
g_k(\tau) \geq g_k(\theta_{k-1}) + \alpha t ||\nabla g_{k-1}(\theta_{k-1})||_2^2 - (N+M) M t \rho^{k+k_0}
$$
may break the monotonicity of the sequence $\{g(\theta_k)\}_{k=0}^\infty$ which would have been used to simplify the convergence analysis in the classical case; and the extra check condition $\tau\in \Theta$ ($\tau$ depends on $k$) forces us
to seek uniform control
\cob
(over all $k$) of the time used to ensure the validity of this condition in each iteration.
In the remainder of this section,
\cob
we deal with these problems and examine the convergence behavior of Algorithm~\ref{algo1}. In a nutshell, we will prove that our algorithm converges exponentially in time under some strong concavity assumptions.

Note that the variable $k$ as in Algorithm~\ref{algo1} actually records the number of times that Step $1$ has been executed at the present moment. To facilitate the analysis of our algorithm, we will put it into an equivalent form, where an additional variable $n$ is used to record the number of times that Step $2$ has been executed.

Below is Algorithm~\ref{algo1} rewritten with the additional variable $n$.

\begin{algorithm} \label{algo2} (An equivalent form of Algorithm~\ref{algo1})

\textbf{Step $0$.} Choose $k_0$ such that Lemma \ref{initial} $(a)$-$(c)$ hold. Set $n=0, k=0, \hat{g}_0=g_{0}=f_{k_0}$, and choose $\alpha\in (0,0.5), \beta\in (0,1)$ and
$\hat{\theta}_0\in \Theta^\circ$
such that $\hat{\theta}_0\in B_{k_0}$ and $\nabla \hat{g}_0(\hat{\theta}_0)\neq 0$.
\cob

\textbf{Step $1$.} Increase $k$ by $1$, and set $t=1$, $g_k=f_{k_0+k}$.

\textbf{Step $2$.} Increase $n$ by $1$.  If $\nabla \hat{g}_{n-1} (\hat{\theta}_{n-1})=0$, set
\begin{equation} \label{case 1 of y}
\tau=\hat{\theta}_{n-1}+t\nabla \hat{g}_{n-1} (\hat{\theta}_{n-1}+\rho^{k+k_0}\mathbf{1});
\end{equation}
otherwise, set
\begin{equation} \label{case 2 of y}
\tau = \hat{\theta}_{n-1} + t \nabla \hat{g}_{n-1}(\hat{\theta}_{n-1}).
\end{equation}
If $\tau \not \in \Theta^\circ$ or
\begin{equation}
g_k(\tau) < g_k(\hat{\theta}_{n-1}) + \alpha t ||\nabla \hat{g}_{n-1}(\hat{\theta}_{n-1})||_2^2 - (N+M) M t \rho^{k+k_0}, \label{increasing condition2}
\end{equation}
then set $\hat{\theta}_n=\hat{\theta}_{n-1}, \hat{g}_n=\hat{g}_{n-1}, t=\beta t$ and go to Step $2$; otherwise, set $\hat{\theta}_n=\tau, \hat{g}_n=g_{k}$ and go to Step $1$.
\end{algorithm}

\begin{rem} \label{hat-no-hat}
Let $n_{0}=0$, and for any $k \geq 1$, recursively define
$$
n_k = \inf\{n>n_{k-1}: \hat{\theta}_n\neq \hat{\theta}_{n-1}\}.
$$
Then, one verifies that for any $k \geq 0$, it holds true that $\hat{\theta}_{n_k}=\theta_k$, $\hat{g}_{n_k}=g_k=f_{k+k_0}$ and moreover, $\hat{\theta}_l = \hat{\theta}_{l+1}$, $\hat{f}_l = \hat{f}_{l+1}$ for any $l$ with $n_{k-1} \leq l \leq n_k-1$, which justify the equivalence between Algorithm~\ref{algo1} and Algorithm~\ref{algo2}.
\end{rem}

The following theorem establishes the exponential convergence of Algorithm~\ref{algo2} with respect to $n$.
\begin{thm} \label{exponential-convergence-1}
Suppose, as in (\ref{little-m}) and (\ref{unique-maximum}), that the strongly concave function $f$ achieves its unique maximum in $\Theta^\circ$. Then there exist $\hat{M}>0$ and $0<\hat{\xi}<1$ such that for all $n\geq 0$,
\begin{equation} \label{hat-conv-rate}
|\hat{g}_n(\hat{\theta}_n)-f(\theta^*)|\leq \hat{M} \hat{\xi}^n,
\end{equation}
where $\hat{g}_n(\hat{\theta}_n)$ is obtained by executing Algorithm \ref{algo2}.
\end{thm}

\begin{proof}

For simplicity, we only deal with the case {\small$\nabla \hat{g}_{n-1} (\hat{\theta}_{n-1}) \neq 0$}  in Step 2 of Algorithm~\ref{algo2} (and therefore (\ref{case 2 of y}) is actually executed), since the opposite case follows from a similar argument by replacing $\hat{\theta}_{n-1}$ with $\hat{\theta}_{n-1}+\rho^{k+k_0} \mathbf{1}$.

Let $T_1(k)$ denote the smallest non-negative integer $p$ such that
\begin{equation} \label{T_1_in_the_first_algo}
\hat{\theta}_{n_{k-1}}+\beta^{p} \nabla \hat{g}_{n_{k-1}} (\hat{\theta}_{n_{k-1}}) \in \Theta^\circ,
\end{equation}
$T(k)$ denote the smallest non-negative integer $q$ such that $q \geq T_1(k)$ and
\begin{equation}
 g_k(\hat{\theta}_{n_{k-1}}+\beta^{q} \nabla \hat{g}_{n_{k-1}} (\hat{\theta}_{n_{k-1}})) \geq g_k(\hat{\theta}_{n_{k-1}})+ \alpha \beta^q ||\nabla \hat{g}_{n_{k-1}}(\hat{\theta}_{n_{k-1}})||_2^2-(N+M)M \beta^q \rho^{k+k_0}. \notag
\end{equation}
Note that the well-definedness of $T_1(k)$ and $T_k$ follows from the observation that if (\ref{T_1_in_the_first_algo}) holds for some non-negative integer $p$, then it also holds for any integer $p' > p$.
Adopting these definitions, we can immediately verify that 
$$
T(k)=n_k-n_{k-1},
$$
which corresponds to the number of times Step $2$ (of Algorithm \ref{algo1}) has been executed to obtain $\hat{\theta}_{n_{k}}$ from $\hat{\theta}_{n_{k-1}}$.

The remainder of the proof consists of the following three steps.

\textbf{Step 1: Uniform boundedness of $T(k)$}. In this step, we show that there exists $A \geq 0$
such that, for all $k$,
\cob
$T(k)\leq A$.

Since $\Theta^\circ$ is open and $\hat{\theta}_{0}\in \Theta^\circ$, we have $T_1(k)<\infty$ for any $k\geq 0$. Note that we haven't show that $T_1(k)$ is uniformly bounded as this stage.

For any $q\geq T_1(k)$, letting
$$
\tau=\hat{\theta}_{n_{k-1}}+\beta^q \nabla \hat{g}_{n_{k-1}} (\hat{\theta}_{n_{k-1}}),
$$
we deduce that $\tau\in \Theta^\circ$ and both $f_k(\tau)$ and $f(\tau)$ are well-defined.
Recalling from (\ref{succeq_M}) that
$$\nabla^2 g_{k}(\theta)=\nabla^2 f_{k+k_0}(\theta) \succeq -M \mathbb{I}_d$$
for any $k\geq 0$ and any $\theta\in \Theta^\circ$,
we derive from the Taylor series expansion that
{\small\begin{align} \label{g_k(tau)}
g_k(\tau) &= g_k(\hat{\theta}_{n_{k-1}})+\beta^{q} \nabla g_{k}(\hat{\theta}_{n_{k-1}})^T \nabla \hat{g}_{n_{k-1}}(\hat{\theta}_{n_{k-1}})+\frac{\beta^{2q}}{2}  \nabla \hat{g}_{n_{k-1}}(\hat{\theta}_{n_{k-1}})^T \nabla^2 g_k(\tilde{\theta}_k)  \nabla \hat{g}_{n_{k-1}}(\hat{\theta}_{n_{k-1}}) \notag\\
& \geq g_k(\hat{\theta}_{n_{k-1}})+\beta^{q} \nabla g_{k}(\hat{\theta}_{n_{k-1}})^T \nabla \hat{g}_{n_{k-1}}(\hat{\theta}_{n_{k-1}})-\frac{M \beta^{2 q} }{2} ||\nabla \hat{g}_{n_{k-1}}(\hat{\theta}_{n_{k-1}})||_2^2,
\end{align}
}
\!\!where $\tilde{\theta}_k\in \Theta^\circ$. According to (\ref{C-rho}), we have
{\small\begin{align*}
&\nabla g_{k}(\hat{\theta}_{n_{k-1}})^T \nabla \hat{g}_{n_{k-1}}(\hat{\theta}_{n_{k-1}})\\
&=\nabla \hat{g}_{n_{k-1}}(\hat{\theta}_{n_{k-1}})^T \nabla \hat{g}_{n_{k-1}}(\hat{\theta}_{n_{k-1}})+(\nabla g_{k}(\hat{\theta}_{n_{k-1}})^T \nabla \hat{g}_{n_{k-1}}(\hat{\theta}_{n_{k-1}})-\nabla \hat{g}_{n_{k-1}}(\hat{\theta}_{n_{k-1}})^T \nabla \hat{g}_{n_{k-1}}(\hat{\theta}_{n_{k-1}})) \\
&\geq ||\nabla \hat{g}_{n_{k-1}}(\hat{\theta}_{n_{k-1}})||_2^2-N\rho^{k+k_0}||\nabla \hat{g}_{n_{k-1}} (\hat{\theta}_{n_{k-1}})||_2.
\end{align*}
}
\!\!This, together with (\ref{g_k(tau)}), implies
{\small\begin{align*}
 g_k(\tau)\geq & g_k(\hat{\theta}_{n_{k-1}})+\beta^{q} ||\nabla \hat{g}_{n_{k-1}}(\hat{\theta}_{n_{k-1}})||_2^2 -\frac{M \beta^{2 q}}{2} ||\nabla \hat{g}_{n_{k-1}}(\hat{\theta}_{n_{k-1}})||_2^2 - N \beta^q \rho^{k+k_0} ||\nabla \hat{g}_{n_{k-1}}(\hat{\theta}_{n_{k-1}})||_2\\
                  \geq& g_k(\hat{\theta}_{n_{k-1}})+\beta^{q} ||\nabla \hat{g}_{n_{k-1}}(\hat{\theta}_{n_{k-1}})||_2^2-\frac{M \beta^{2 q}}{2} ||\nabla \hat{g}_{n_{k-1}}(\hat{\theta}_{n_{k-1}})||_2^2- N M \beta^q \rho^{k+k_0},
\end{align*}
}
\!\!where the last inequality follows from (\ref{m-M}).
Note that for any non-negative integer {\small$q \geq -\log M/\log \beta$}, we have
$$
\beta^{q}-\frac{M \beta^{2q}}{2} \geq \frac{1}{2} \beta^{q}> \alpha \beta^{q},
$$
which immediately implies that (\ref{increasing condition2}) fails; in other words, for any non-negative integer $q\geq T_1(k)$, we have
\begin{equation}
 g_k(\hat{\theta}_{n_{k-1}}+\beta^{q} \nabla \hat{g}_{n_{k-1}} (\hat{\theta}_{n_{k-1}})) \geq g_k(\hat{\theta}_{n_{k-1}})+ \alpha \beta^q ||\nabla \hat{g}_{n_{k-1}}(\hat{\theta}_{n_{k-1}})||_2^2-(N+M)M \beta^q \rho^{k+k_0} \notag
\end{equation}
as long as $q\geq -\log M/\log \beta$.
It then follows that for any integer $k\geq 1$, $T(k)$ can be bounded as
\begin{align} \label{T(k)_two_cases}
T(k)\leq
\begin{cases}
A_2 \qquad &\mbox{if $T_1(k)\leq A_2$} \\
T_1(k)  &\mbox{if $T_1(k) >A_2$},
\end{cases}
\end{align}
where $A_2\triangleq \max\{0, -\log M/ \log \beta+1\}$ is a constant independent of $k$.
Now, to prove the uniform boundedness of $T(k)$, what remains is to show that there exists $A_1 \geq 0$ such that for all $k$,
$T_1(k)\leq A_1$.

When $T_1(k)\leq A_2$, we can simply set $A_1=A_2$ and deduce $T_1(k)\leq A_1$. When $T_1(k)> A_2$, recalling that $T_1(k)< \infty$ for any $k$, we can always find $\{\hat{\theta}_{n_k}\}_{k=0}^\infty$ such that for all $k\geq 0$, $\hat{\theta}_{n_k}\in \Theta^\circ$
and
\begin{equation} \label{incc}
g_k(\hat{\theta}_{n_{k}}) \geq g_k(\hat{\theta}_{n_{k-1}})+ \alpha \beta^{T_1(k)} ||\nabla \hat{g}_{n_{k-1}}(\hat{\theta}_{n_{k-1}})||_2^2-(N+M)M \beta^q \rho^{k+k_0}.
\end{equation}
Note that (\ref{C-rho}) and (\ref{incc}) together imply that
\cob
\begin{align}
g_{k}(\hat{\theta}_{n_k}) &\geq g_{k-1} (\hat{\theta}_{n_{k-1}})-(N+M)M \rho^{k+k_0}-N\rho^{k+k_0}, \notag
\end{align}
from which, by induction on $k$, we further obtain that
\begin{align*}
g_{k}(\hat{\theta}_{n_k}) &\geq g_{0}(\hat{\theta}_{0})-\sum_{i={0}}^{k-1} \left[(N+M) M \rho^{i+k_0+1}+N\rho^{i+k_0+1} \right] \\
& \geq g_{0}(\hat{\theta}_{0}) - \left[\frac{(N+M)M \rho^{k_0+1}}{1-\rho}+\frac{N \rho^{k_0+1}}{1-\rho} \right].
\end{align*}
It then follows that for all $k\geq 0$,
\begin{align} \label{f_k_0}
g_{0}(\hat{\theta}_{n_k})&\geq g_{0} (\hat{\theta}_{k_0}) -\left[\frac{(N+M)M\rho^{k_0+1}}{1-\rho}+\frac{N\rho^{k_0+1}}{1-\rho} \right] -\sum_{i=1}^{k} N \rho^{i+k_0} \notag \\
&\geq g_{0} (\hat{\theta}_{k_0}) -\left[\frac{(N+M)M\rho^{k_0+1}}{1-\rho}+\frac{2 N\rho^{k_0+1}}{1-\rho} \right]  \notag \\
&\geq g_{0} (\hat{\theta}_{k_0}) - \frac{\delta}{8},
\end{align}
where the last inequality follows from Lemma \ref{initial} $(a)$.
Now, letting $y_0, B_{k_0}$ and $C_{k_0}$ be defined as in Lemma \ref{initial}, we infer from (\ref{f_k_0}) and Lemma \ref{initial} $(c)$ that $\{\hat{\theta}_{n_k}\}_{k=0}^\infty \subseteq C_{k_0} \subseteq \Theta^\circ.$ Hence, for any non-negative integer
{\small$p\geq \log (dist(C_{k_0}, \partial\Theta)/M)/\log \beta$},
\cob
we have {\small$\hat{\theta}_{n_{k-1}}+\beta^{p} \nabla \hat{g}_{n_{k-1}} (\hat{\theta}_{n_{k-1}}) \in \Theta^\circ$} and it then follows that $T_1(k)\leq A_1$, where $A_1$ is defined as 
\begin{align} \label{A_1_bound}
A_1\triangleq\max \left\{0, \frac{\log (dist(C_{k_0}, \partial\Theta)/M )}{\log \beta}+1 \right\}.
\end{align}
\cob
Finally, it follows from (\ref{T(k)_two_cases}) and (\ref{A_1_bound}) that
\begin{align} \label{T(k)}
T(k)\leq A\triangleq \max\{A_1, A_2\},
\end{align}
as desired.

\textbf{Step 2: Exponential convergence of $\{f(\hat{\theta}_{n_k})\}$}.
First of all, for any $k\geq k_0$, from the definition of $T(k)$,
we have
$$
g_k(\hat{\theta}_{n_{k}}) \geq g_k(\hat{\theta}_{n_{k-1}})+\alpha \beta^{T(k)} ||\nabla g_{k-1}(\hat{\theta}_{n_{k-1}})||_2^2- (N+M) M \beta^{T(k)} \rho^{k+k_0}.
$$
\cob
Using (\ref{C-rho}), (\ref{m-M}) and
\cob
the definition of $\{\hat{g}_{n_k}\}_{k=0}^\infty$,
we now write
\cob
$$
f(\hat{\theta}_{n_k}) \geq f(\hat{\theta}_{n_{k-1}})+\alpha \beta^{T(k)} ||\nabla f(\hat{\theta}_{n_{k-1}})||_2^2- [(N+M)M\beta^{T(k)}+2 N+2 N M\rho] \rho^{k+k_0}.
$$
\cob
According to (\ref{little-m}), we have
$$
f(\theta^*) \leq f(\hat{\theta}_{n_{k-1}}) + \nabla f(\hat{\theta}_{n_{k-1}})^T (\theta^*-\hat{\theta}_{n_{k-1}}) -\frac{m}{2} ||\theta^*-\hat{\theta}_{n_{k-1}}||_2^2,
$$
which,
coupled with some straightforward estimates,
\cob
yields
$$
2m (f(\theta^*)-f(\hat{\theta}_{n_{k-1}})) \leq ||\nabla f(\hat{\theta}_{n_{k-1}})||_2^2.
$$
It then follows that
{\small\begin{align} \label{f}
&f(\theta^*)-f(\hat{\theta}_{n_k}) \notag \\
&\leq f(\theta^*)-f(\hat{\theta}_{n_{k-1}}) - \alpha \beta^{T(k)} ||\nabla f(\hat{\theta}_{n_{k-1}})||_2^2 +[(N+M)M\beta^{T(k)}+2 N+2 N M\rho] \rho^{k+k_0} \notag \\
		       &\leq  (1-2 m \alpha \beta^{T(k)})  (f(\theta^*)-f(\hat{\theta}_{n_{k-1}})) +[(N+M)M+2 N+2 N M\rho] \rho^{k+k_0} \notag \\
&\overset{(d)}\leq \left(1-\min\left\{ 2 m \alpha \beta^{A_1}, 2 m \alpha \beta^{A_2} \right\}\right)(f(\theta^*)-f(\hat{\theta}_{n_{k-1}})) +\left(\frac{NM+M^2+2N}{\rho}+2 N M\right) \rho^{k+k_0+1} \notag\\
              & = \eta (f(\theta^*)-f(\hat{\theta}_{n_{k-1}})) + \gamma_k,
\end{align}}
\cob
\!\!where
{\small \begin{align*}
\eta =1-\min\left\{2m\alpha, \frac{dist(C_{k_0},\partial \Theta)}{M}2m\alpha\beta , \frac{2m\alpha \beta}{M}\right\}, \quad\gamma_k=\left(\frac{NM+M^2+2N}{\rho}+2 N M\right)\rho^{k+k_0+1}
\end{align*}}
\cob
\!\!and $(d)$ follows from (\ref{T(k)}). Recursively applying inequality (\ref{f}) and noting that $0<\eta<1$, we infer that there exist $0<\xi<1$ and $M'>0$ such that
\begin{equation} \label{n_k}
f(\theta^*)-f(\hat{\theta}_{n_k})\leq M' \xi^k.
\end{equation}

\textbf{Step 3: Exponential convergence of $\{\hat{g}_n(\hat{\theta}_n)\}$}. In this step, we establish (\ref{hat-conv-rate}) and thereby finish the proof.

First, note that for any  positive integer $n\geq 0$, there exists an integer $k'\geq 0$ such that
\begin{align}
n_{k'} \leq n \leq n_{k'+1}, \quad n\leq (k'+1) A, \quad \hat{\theta}_n=\hat{\theta}_{n_{k'}}, \quad \hat{g}_n(\hat{\theta}_n)=\hat{g}_{n_{k'}}(\theta_{n_{k'}}), \notag
\end{align}
where $A$ is defined in (\ref{T(k)}). These four inequalities, together with (\ref{C-rho}) and (\ref{n_k}), imply the existence of $\hat{M}>0$ and $0<\hat{\xi}<1$ such that for any $n\geq 0$,
\begin{align}
|\hat{g}_n(\hat{\theta}_n)-f(\theta^*)| &\leq |\hat{g}_{n_{k'}}(\hat{\theta}_{n_{k'}}) -f(\hat{\theta}_{n_{k'}})|+|f(\hat{\theta}_{n_{k'}})-f(\theta^*)| \notag \\
& \leq N\rho^{k'+k_0} + M' \xi^{k'} \notag \\
&\leq N \rho^{k_0} \rho^{\lfloor n/{A}\rfloor -1} + M' \xi^{\lfloor n/{A} \rfloor -1} \notag \\
&\leq \hat{M} \hat{\xi}^n, \notag
\end{align}
\cob
which completes the proof of the theorem.
\end{proof}

Theorem~\ref{exponential-convergence-1}, together with the uniform boundedness of $T(k)$ established in its proof, immediately implies that Algorithm~\ref{algo1} exponentially converges in $k$. More precisely, we have the following theorem.
\begin{thm} \label{exponential-convergence-2}
For a strongly concave function $f$ whose unique maximum is achieved in $\Theta^\circ$, there exist $\tilde{M}>0$ and $0<\tilde{\xi}<1$ depending on $m,M,N$ and $\rho$ such that for all $k$,
\begin{equation} \label{tilde-conv-rate}
|g_k(\theta_k)-f(\theta^*)|\leq \tilde{M} \tilde{\xi}^k,
\end{equation}
where $g_k(\theta_k)$ is defined as in Algorithm \ref{algo2}.
\end{thm}

\subsection{Applications of Algorithm~\ref{algo1}} \label{applications}

In this section, we discuss
some applications of
\cob
Algorithm \ref{algo1} in information theory.

Consider a finite-state channel satisfying (\ref{channel-model}.$a$)-(\ref{channel-model}.$c$) and assume that all the matrices in $\Pi_{F, \delta}$ are
analytically
parameterized
by $\theta \in \Theta^\circ$, where $\Theta$ is a compact convex subset of $\mathbb{R}^d$, $d \in \mathbb{N}$. Setting
$$
f(\theta)=I(X(\theta); Y(\theta))
$$
and
$$
f_k(\theta) =
H(X_2(\theta)|X_1(\theta))
+H(Y_{k+1}(\theta)|Y_1^{k}(\theta))-H(X_{k+1}(\theta), Y_{k+1}(\theta)|X_1^k(\theta), Y_1^k(\theta)),
$$
we derive from \cite{han11} that (\ref{C-rho}) holds.
So, when $f(\theta)$ is strongly concave
with respect to $\theta$ (this may hold true for some special channels, see, for example, \cite{hm09b} and \cite{LiHan}) as in (\ref{little-m}), our algorithm applied to $\{f_k(\theta)\}_{k=0}^\infty$ converges 
exponentially fast in the number of steps to the maximum value of $f(\theta)$.
\cob
This, together with Theorem~\ref{exponential-convergence-2} and the easily verifiable fact
that the computational complexity of $f_k(\theta)$ is at most exponential in $k$, leads to the conclusion that Algorithm~\ref{algo1}, when applied to $\{f_k(\theta)\}_{k=0}^\infty$ as above, achieves exponential accuracy in exponential time. We now trade exponential time for polynomial time at the expense of accuracy. For any fixed $r \in \mathbb{R}_+$ and any large $k$, choose the largest $l \in \mathbb{N}$ such that $k=\lceil r \log l \rceil$. Substituting this into (\ref{tilde-conv-rate}), we have
$$
|f_{\lceil r \log l \rceil}(\theta_{\lceil r \log l \rceil})-f(\theta^*)|\leq \tilde{M} l^{r \log \tilde{\xi}}.
$$
In other words, as summarized in the following theorem, we have shown that Algorithm~\ref{algo1}, when used to compute the channel capacity as above, achieves polynomial accuracy in polynomial time.
\begin{thm} \label{general-channel}
For a general finite-state channel satisfying (\ref{channel-model}.a)-(\ref{channel-model}.c) and parameterized as above, if $I(X(\theta); Y(\theta))$ is strongly concave with respect to $\theta \in \Theta$ and achieves its unique maximum in $\Theta^\circ$, then there exists an algorithm computing its fixed order Markov capacity that achieves polynomial accuracy in polynomial time.
\end{thm}

In the following, we show that for certain special families of finite-state channels, we do get a stronger convergence result than that in Theorem \ref{general-channel}. In particular, for the following two examples, Algorithm~\ref{algo1} can be used to compute the channel capacity, achieving exponential accuracy in polynomial time.

\subsubsection{A noisy channel with one state} \label{Section BEC}

In this section, we consider the Markov capacity of a binary erasure channel (BEC) under the $(1, \infty)$-RLL constraint. This channel can be mathematically characterized by the input-output equation
\begin{equation} \label{BEC}
Y_n=X_n \cdot E_n,
\end{equation}
where $\{X_n\}_{n=1}^\infty$ is the input stationary Markov chain taking values in $\{1,2\}$ such that $\{22\}$ is a {\it forbidden set} (see, e.g.,~\cite{lm95}),
and $\{E_n\}_{n=1}^\infty$ is an i.i.d. process taking values in $\{0,1\}$ with
$$
P(E_n=0)=\varepsilon, \quad P(E_n=1)=1-\varepsilon
$$
for $0<\varepsilon<1$. Here we note that the BEC given above can be viewed as a degenerate finite-state channel with only one state. In the following, we will compare the channel capacity when $\{X_n\}_{n=1}^\infty$ is a first-order stationary Markov chain with the capacity when $\{X_n\}_{n=1}^\infty$ is a second-order stationary Markov chain. In particular, Algorithm \ref{algo1} will be used to evaluate the first-order Markov capacity, which, compared to a lower bound for the second-order Markov capacity, will lead to the conclusion that higher order memory in the channel input may increase the Markov capacity.

For the first case, suppose that $\{X_n\}_{n=1}^\infty$ is a first-order stationary Markov chain with the transition probability matrix (indexed by $1$, $2$)
$$
\Pi=\left[ \begin{array}{cc}
1-\theta & \theta\\
1 & 0
\end{array} \right]
$$
for $0<\theta<1$.  It has been established in~\cite{LiHan2018} that the mutual information rate $I(X(\theta); Y(\theta))$ of the BEC channel (\ref{BEC}) can be computed as
$$
I(X(\theta); Y(\theta))= (1-\varepsilon)^2 \sum_{l=0}^{\infty} H(X_{l+2}(\theta)|X_1(\theta)) \varepsilon^l,
$$
which is strictly concave with respect to $\theta$. Now, setting $f(\theta)= I(X(\theta); Y(\theta))$, one verifies, through straightforward computation, that
$$
f(\theta) = \lim_{k\rightarrow \infty} f_k (\theta),
$$
where
{\small\begin{align*}
f_0(\theta)= &f_1(\theta)\triangleq (1-\varepsilon)^2 \frac{-\theta \log \theta-(1-\theta) \log (1-\theta)}{1+\theta},\\
f_k(\theta) \triangleq &(1-\varepsilon)^2 \frac{-\theta \log \theta-(1-\theta) \log (1-\theta)}{1+\theta} \\
&+(1-\varepsilon)^2 \sum_{l=2}^{k} \left\{  \frac{1}{1+\theta} H\left(\frac{1-(-\theta)^{l+1}}{1+\theta}\right) \right\} \varepsilon^{l-1} +(1-\varepsilon)^2 \sum_{l=2}^{k} \left\{\frac{\theta}{1+\theta} H\left(\frac{1-(-\theta)^l}{1+\theta}\right) \right\} \varepsilon^{l-1}
\end{align*}}
\!\!for $k\geq 2$ and $H(p)\triangleq-p\log p - (1-p) \log(1-p)$ is the binary entropy function.
In what follows,
assuming $\varepsilon=0.1$,
we will show that Algorithm~\ref{algo1} can be applied to compute the first-order Markov capacity of the channel (\ref{BEC}), i.e., the maximum of $f(\theta)$ over all $\theta \in [0, 1]$.


First of all, we claim that $f(\theta)$ achieves its unique maximum within the interval $[0.25, 0.55]$ and therefore in the interior of $\Theta\triangleq[0.2,0.6]$. To see this, noting that $f_k(\theta)\leq f(\theta)$ for any $\theta$ and through evaluating the elementary function $f_{100}(\theta)$, we have
$$0.442239<\max_{\theta\in [0.25, 0.55]} f_{100}(\theta)<0.442240$$
and therefore
\begin{equation} \label{f100}
\max_{\theta \in [0.25, 0.55]} f(\theta) \geq  0.442239,
\end{equation}
where (\ref{f100}) follows from the fact that $f_k(\theta)$ is monotonically increasing in $k$.
On the other hand, using the stationarity of $\{Y_n\}_{n=1}^\infty$ and the fact that conditioning reduces entropy, we have
\begin{align*}
f(\theta)=I(X(\theta); Y(\theta))= H(Y)-H(\varepsilon) \leq H(Y_3(\theta)|Y_{1}(\theta),Y_2(\theta))-H(\varepsilon),
\end{align*}
where $H(Y)$ is the entropy rate of $\{Y_n\}_{n=1}^\infty$. Then, by straightforward computation, we deduce that
$$
\max_{\theta\in [0,0.25] \cup [0.55,1]} f(\theta) \leq \max_{\theta\in [0,0.25] \cup [0.55,1]} H(Y_3(\theta)|Y_{1}(\theta), Y_2(\theta))-H(\varepsilon) < 0.414483,
$$
which, together with (\ref{f100}), yields
$$
\max_{\theta\in [0,0.25] \cup [0.55,1]} f(\theta) < \max_{\theta \in [0.25, 0.55]} f(\theta),
$$
as desired.

Next, we will verify that (\ref{C-rho}), (\ref{m-M}) and (\ref{little-m}) are satisfied for all $\theta \in [0.2, 0.6]$. 
Note that for $k\geq 2$ we have
$$
f_k(\theta)-f_{k-1}(\theta) = (1-\varepsilon)^2\left[ \frac{1}{1+\theta} H\left(\frac{1-(-\theta)^{k+1}}{1+\theta}\right)+\frac{\theta}{1+\theta} H\left(\frac{1-(-\theta)^k}{1+\theta}\right) \right] \varepsilon^{k-1}.
$$
\cob
This implies that
\cob
for any $k \geq 5$ and any $\theta \in [0.2, 0.6]$,
$$
|f_k(\theta)-f_{k-1}(\theta)|\leq (1-\varepsilon)^2 \varepsilon^{k-1}=8.1\times 0.1^k.
$$
This, together with the easily verifiable fact that
$0.378\leq f_5(\theta)\leq0.443$ for $\theta \in [0.2, 0.6]$,
\cob
further implies that
$$
|f_k(\theta)-f(\theta)|\leq 0.9\times 0.1^k \quad \mbox{and} \quad 0.37 \leq f_k(\theta)\leq 0.45
$$
\cob
for all $k \geq 5$ and $\theta\in [0.2,0.6]$.

Going through similar arguments,
we obtain that, for any $k \geq 13$ and any $\theta \in [0.2, 0.6]$,
\cob
$$
|f_k' (\theta)-f_{k-1}' (\theta)| \leq 72.9 \times 0.1^{k}, \quad |f_k'(\theta)-f'(\theta)|\leq 8.1\times 0.1^k,
$$
and
$$
-0.44\leq f_k'(\theta)\leq 0.76,
$$
and, for any $k \geq 18$ and any $\theta \in [0.2, 0.6]$,
$$
|f_k''(\theta)-f_{k-1}''(\theta)|\leq 370.575\times 0.1^{k}, \quad |f_{k}''(\theta)-f''(\theta)| \leq 41.175\times 0.1^{k},
$$
\cob
and
$$
-5.81 \leq f_k''(\theta)\leq -1.88.
$$
To sum up, we have shown that (\ref{C-rho}) is satisfied with $N=371$ and $\rho=0.1$, (\ref{m-M}) is satisfied with $M=5.81$ and (\ref{little-m}) is satisfied with $m=1.88$.
Under these choices of the constants,
\cob
direct calculation shows that $k_0=18$ is sufficient for Lemma \ref{initial}.
As a result, Algorithm~\ref{algo1} is applicable to the channel (\ref{BEC}).
\cob
Observing that, by its definition, the computational complexity of $f_k(\theta)$ is polynomial in $k$,
we conclude that Algorithm~\ref{algo1} achieves exponential accuracy in polynomial time.

Now, applying Algorithm~\ref{algo1} to the sequence $\{f_k(\theta): k \geq 18\}$ over $\Theta = [0.2, 0.6]$ with $\alpha=0.4$, $\beta=0.9$ and the initial point $\theta_0=0.5$, we obtain that
$$
\theta_{110}\approx 0.395485, \quad f_{110}(\theta_{110})\approx 0.442239.
$$
Furthermore, under the settings given above, $\xi$ and $\eta$ can be chosen such that $\xi=\eta<0.767$. It now follows from (\ref{C-rho}), (\ref{n_k})
and
\cob
$\hat{\theta}_{n_k} = \theta_k$ (see Remark~\ref{hat-no-hat}) that
{\small $$
|f_{110}(\theta_{110})-f(\theta^*)| \leq |f_{110} (\theta_{110})-f(\theta_{110})|+|f(\theta_{110})-f(\theta^*)| \leq 2.621\times 10^{-7},
$$}
\!\!which further implies that when the input is a first-order Markov chain, the capacity of the BEC channel (\ref{BEC}) can be bounded as
\begin{equation} \label{BECcapacity}
0.4422382\leq f(\theta^*)\leq 0.4422398.
\end{equation}

We now consider the case when the input is a second-order stationary Markov chain, whose transition probability matrix (indexed by $11, 12$ and $21$ only since $22$ is prohibited by the $(1,\infty)$-RLL constraint) is given by
$$
\left[\begin{array}{ccc}
p&1-p&0\\
0&0&1\\
q&1-q&0
\end{array}\right],
$$
where $0<p, q<1$. For this case,
from the Birch lower bound (see, e.g., Lemma 4.5.1 of \cite{co06}), we have
$$H(Y_6|Y_5,Y_4,Y_3, X_2, X_1)-H(\varepsilon)\leq H(Y)-H(\varepsilon)=I(X;Y).$$
It can then be verified by direct computation that, when $p=0.597275$ and $q=0.614746$,
\cob
$$
H(Y_6|Y_5,Y_4,Y_3, X_2, X_1)-H(\varepsilon)=0.442329,
$$
which is a lower bound on the second-order Markov capacity yet strictly larger than the
upper bound
\cob
on the first-order Markov capacity given in (\ref{BECcapacity}). Hence we can draw the conclusion that for the BEC channel with Markovian inputs under the $(1,\infty)$-RLL constraint, an increase of the Markov order of the input process from $1$ to $2$ does increase the channel capacity.

\subsubsection{A noiseless channel with two states} \label{noiseless}

In this section, we consider a noiseless finite-state channel with two channel states, for which we show that Algorithm \ref{algo1} can be applied to show that higher order memory can yield larger Markov capacity.

More precisely, the channel input $\{X_n\}_{n=1}^\infty$ is a first-order stationary Markov chain
taking values from
the alphabet $\mathcal{A}=\{0, 1\}$ and,
\cob
except at time $0$, the channel state $\{S_n\}_{n=1}^\infty$ is determined by the channel input, that is, $S_n=X_n$, $n=1, 2, \dots$. The channel is characterized by the following input-output equation:
\begin{equation} \label{Brian}
Y_n=\phi(S_{n-1}, X_n), \quad n = 1, 2, \dots,
\end{equation}
where $\phi$ is a deterministic function with $\phi(0, 0)=1, \phi(0, 1)=0, \phi(1, 0)=0$ and $\phi(1, 1)=0$. Note that $\phi$ naturally induces a sliding block code that maps the full $\mathcal{A}$-shift $\mathcal{S}$ to the shift of finite type $\mathcal{S}_{\mathcal{F}}$, where the forbidden set $\mathcal{F}$ is $\{101\}$. It can be readily verified that the Shannon capacity of (\ref{Brian}) is equal to its stationary capacity~\cite{Gray2011}, which can be computed as the largest eigenvalue of the adjacency matrix of the $3$rd higher block shift of $\mathcal{S}_{\mathcal{F}}$ and is approximately equal to $0.562399$ (see Chapter $4$ and $13$ of \cite{lm95} for more details). In what follows, we will focus on the Markov capacity of (\ref{Brian}); more specifically, we will compute the Markov capacity when the input $\{X_n\}_{n=1}^\infty$ is an i.i.d. process and a first-order stationary Markov chain, which will be compared with the Shannon capacity.

It can be easily verified that the mutual information rate of (\ref{Brian}) can be computed as
$$
I(X; Y)=\lim_{k \rightarrow \infty}H(Y_{k+1}|Y_{1}^{k})-\frac{1}{k} H(Y_1^k|X_1^k)=\lim_{k\rightarrow \infty}H(Y_{k+1}|Y_{1}^{k})=H(Y).
$$
\cob
When $\{X_n\}_{n=1}^\infty$ is a stationary Markov chain, the output $\{Y_n\}_{n=1}^\infty$ is a hidden Markov chain with an unambiguous symbol whose entropy rate can be computed by the following formula \cite{gm05}:
\begin{align} \label{unambiguous_formula}
H(Y)=\sum_{n=1}^\infty P(Y_1^n=(1,\underbrace{0,\dots,0}_{n-1})) H(Y_{n+1}|Y_1^n=(1, \underbrace{0,\dots,0}_{n-1})).
\end{align}
This formula will play a key role in our analysis, detailed below.

We first consider the degenerated case that $\{X_n\}_{n=1}^\infty$ is an i.i.d. process. Letting $\theta$ denote $P(X_1=0)$, we note that the Markov chain $\{(X_{n-1}, X_{n})\}_{n=2}^\infty$ has the following transition probability matrix (indexed by $00, 01, 10, 11$)
\begin{equation} \label{rank-1}
\left[\begin{array}{cccc}
\theta&1-\theta&0&0\\
0&0&\theta&1-\theta\\
\theta&1-\theta&0&0\\
0&0&\theta&1-\theta \notag
\end{array}
\right],
\end{equation}
whose left eigenvector
corresponding to the largest eigenvalue
\cob
is
$$
(\pi_1(\theta), \pi_2(\theta), \pi_3(\theta), \pi_4(\theta))=(\theta^2, \theta(1-\theta), \theta(1-\theta), (1-\theta)^2).
$$
Using (\ref{unambiguous_formula}), we have
$$
H(Y)=-\sum_{l=0}^{\infty} \pi_1(\theta) \mathbf{r} (B_{\theta})^l \mathbf{1} \log \frac{\mathbf{r} (B_\theta)^l \mathbf{1}}{\mathbf{r} (B_\theta)^{l-1} \mathbf{1}}-\sum_{l=0}^{\infty} \pi_1(\theta) \mathbf{r} (B_\theta)^{l-1} \mathbf{c} \log \frac{\mathbf{r} (B_\theta)^{l-1} \mathbf{c}}{\mathbf{r} (B_\theta)^{l-1} \mathbf{1}},
$$
where $\mathbf{r}=(1-\theta, 0, 0)$, $\mathbf{c}=(0, \theta, 0)^T$, $\mathbf{1}=(1, 1, 1)^T$,
$$
B_\theta=\left[\begin{array}{ccc}
0&\theta&1-\theta\\
1-\theta&0&0\\
0&\theta&1-\theta
\end{array}\right],
$$
and both $\mathbf{r} (B_\theta)^{-1} \mathbf{1}, \mathbf{r} (B_\theta)^{-1} \mathbf{c}$ should be interpreted as $1$.

Setting $f(\theta)=H(Y)$, we note that
$$
f(\theta) = \lim_{k \to \infty} f_k(\theta),
$$
where
$$
f_k(\theta)=-\sum_{l=0}^{k} \pi_1(\theta) \mathbf{r} (B_\theta)^l \mathbf{1} \log \frac{\mathbf{r} (B_\theta)^l \mathbf{1}}{\mathbf{r} (B_\theta)^{l-1} \mathbf{1}}-\sum_{l=0}^{k} \pi_1(\theta) \mathbf{r} (B_\theta)^{l-1} \mathbf{c} \log \frac{\mathbf{r} (B_\theta)^{l-1} \mathbf{c}}{\mathbf{r} (B_\theta)^{l-1} \mathbf{1}}, \quad k \geq 0.
$$
Similarly as in the previous section, we can show that
$$
\max_{\theta \in [0, 0.41] \cup [0.89, 1]} f(\theta) < \max_{\theta \in [0.41, 0.89]} f(\theta),
$$
which means that $f(\theta)$ will achieve its maximum within the interior of $[0.4, 0.9]$. Moreover, through tedious but similar evaluations as in the previous example,
we can choose
(below, rather than a constant, $N$ is a polynomial in $k$, but the proof of Theorem~\ref{exponential-convergence-1} carries over almost verbatim)
\cob
$$
k_0=120, \quad N=(374.945k^2+6207.73k+46587.2),\quad \rho=0.875, \quad m=1.2, \quad M=10.37.
$$
Though the function $f(\theta)$ is not concave near $\theta=0$, tedious yet straightforward computation indicates that $f''(\theta)\leq f_{120}''(\theta)+N\rho^{120}<0$ for any $\theta\in[0.4, 0.9]$, which immediately implies that
$f(\theta)$ is strongly concave
within the interior of the interval $[0.4, 0.9]$. Then, similarly as in
\cob
Section~\ref{Section BEC},
one verifies that, when applied to the channel in (\ref{Brian}),
\cob
Algorithm~\ref{algo1} achieves exponential accuracy in polynomial time.

Letting $\alpha=0.4, \beta=0.9$, we apply our algorithm to the sequence $\{f_k(\theta): k \geq 120\}$ with $\Theta\triangleq[0.4, 0.9]$, $\theta_0=0.5$, $\eta=\xi=0.901061$, and we obtain that
$$
\theta_{450}\approx 0.6257911, \quad f_{450}(\theta_{450})\approx 0.4292892.
$$
Now from (\ref{C-rho}), (\ref{n_k}) and the fact that $\hat{\theta}_{n_k} = \theta_k$, we conclude
{\small $$
|f_{450}(\theta_{450})-f(\theta^*)| \leq |f_{450} (\theta_{450})-f(\theta_{450})|+|f(\theta_{450})-f(\theta^*)| \leq 0.0001745,
$$}
\!\!which further implies
\begin{equation} \label{Symbolic dynamics capacity}
0.4291146 \leq f(\theta^*)\leq 0.4294638
\end{equation}
for the i.i.d. case.

Now, we consider the case that $\{X_n\}_{n=1}^{\infty}$ is a genuine first-order stationary Markov process, and assume the Markov chain $\{(X_{n-1}, X_n)\}_{n=2}^\infty$ has the following transition probability matrix (indexed by $00, 01, 10, 11$)
$$
\begin{pmatrix}
p & 1-p & 0 & 0\\
0 & 0 & q & 1-q \\
p & 1-p & 0 & 0\\
0 & 0 & q & 1-q
\end{pmatrix},
$$
where $0< p, q<1$. Again, straightforward computation shows that for $p=0.674521, q=0.595176$, $H(Y_4| Y_3, X_2, X_1)$ is approximately $0.513259$, which gives a lower bound on $H(Y)$. Comparing this lower bound with the upper bound in (\ref{Symbolic dynamics capacity}), we conclude that the capacity is increased when increasing the Markov order of the input from $0$ to $1$.



\section{The Second Algorithm: without Concavity} \label{without concavity}

In this section, we consider the optimization problem (3) for the case when $f$ may {\bf not} be concave.

For a non-convex optimization problem with a continuously differentiable target function $f$ and a bounded domain, conventionally there are two major methods for finding its solution: the Frank-Wolfe method and the method through the \L ojasiewicz inequality (see, e.g., \cite{ab06}). However,
both of these methods
\cob
in general tend to fail in our setting: for the Frank-Wolfe method, the computation for finding the feasible ascent direction and the verification of the relevant gradient condition (which is necessary for the convergence of this method) both depend on the existence of
an exact formula for $\nabla f$
\cob
and a tractable description of $\Theta$, which is however not available in our case; on the other hand, due to the fact that our target function
is the limit of
\cob
a sequence of approximating functions, the method through the \L ojasiewicz inequality necessitates a ``uniform'' version of the \L ojasiewicz inequality over all sequences of approximating functions, which does not seem to hold true in our setting.

Motivated by Algorithm \ref{algo1}, we propose in the following our second algorithm to efficiently solve the optimization problem (3) whose target function may not be concave. Except for using the sequence $\{\nabla f_k\}_{k=0}^\infty$ as the ascent direction in each iteration, an additional check condition is proposed for the choice of the step size. This check condition is chosen carefully to ensure an appropriate pace for the decay of $\nabla f_k$, which turns out to be crucial for the convergence of this algorithm.

Similarly as in Section \ref{algorithm}, we need the following lemma before presenting our second algorithm.
\begin{lem} \label{initial2}
Assume the function $f$ has $s$ stationary points $\{\theta_i^*\}_{i=1}^s$
which are all contained in $\Theta^\circ$, and that $f$ achieves its maximum in $\Theta^\circ$.
\cob
If, for each $k$,
\cob
$f_k$ also has finitely many stationary points which are all contained in $\Theta^\circ$, then there exists a non-negative integer $k_0$ such that
\begin{enumerate}
\item[(a)] $\rho^{1/3}+\rho^{2 k_0/3}<1$ and $\displaystyle\frac{2 N \rho^{k_0}}{1-\rho}\leq \frac{\delta}{8}$, where $\delta\triangleq \max\limits_{\theta^*_i: 1\leq i\leq s } f(\theta_i^*) - \max\limits_{\theta\in \partial \Theta} f(\theta)>0$;
\item[(b)] There exists $y_0\in \mathbb{R}$ such that for any fixed $b$ with $0<b<1$, we have
$$
\emptyset\subsetneq  B_{k_0}\subseteq C_{k_0}\subseteq \Theta^\circ,\quad A_{k_0} \cap B_{k_0}\neq \emptyset \quad \mbox{and} \quad dist(C_{k_0}, \partial\Theta)>0,
$$
where
\begin{align*}
A_{k_0}&\triangleq \left\{ x\in \Theta^\circ: ||\nabla f_{k_0} (x)||_2\geq \frac{2N \rho^{k_0/3}}{1-b}\right\},\\
B_{k_0}&\triangleq \{x\in \Theta: f_{k_0}(x) \geq y_0\},\\
C_{k_0}&\triangleq \left\{x\in\Theta: f_{k_0}(x)\geq y_0- \frac{\delta}{8}\right\}.
\end{align*}
\end{enumerate}
\end{lem}

\begin{proof}
By replacing what was assumed to be the unique maximum of $f$
\cob
with $\max_{\theta^*_i: 1\leq i\leq s } f(\theta_i^*)$, a similar argument as in the proof of
Lemma \ref{initial}($a$)
\cob
yields that there exists $y_0<y^*-\frac{\delta}{4}$ such that for all sufficiently large $k$, $\emptyset\subsetneq  B_{k}\subseteq C_{k}\subseteq \Theta^\circ$ and $dist(C_{k}, \Theta^c)>0$,
where
\begin{align*}
y^*=\max_{\theta^*_i: 1\leq i\leq s } f(\theta_i^*), \quad B_{k}\triangleq \{x\in \Theta: f_{k}(x) \geq y_0\}\quad \mbox{and} \quad C_{k}\triangleq \left\{x\in\Theta: f_{k}(x)\geq y_0- \frac{\delta}{8}\right\}.
\end{align*}
Now, for any $k$ and any fixed $0<b<1$, let
$$A_{k}\triangleq \left\{ x\in \Theta^\circ: ||\nabla f_{k} (x)||_2\geq \frac{2N \rho^{k/3}}{1-b}\right\}.$$
We claim that for large enough $k$, $A_k \cap B_k\neq \emptyset$. To see this, define
$$
D_k\triangleq \left\{x\in \Theta^\circ: ||\nabla f(x)||_2\geq \frac{2N \rho^{k/3}}{1-b}+N\rho^k\right\} \quad \mbox{and} \quad B'\triangleq \left\{x\in \Theta: f(x)\geq y_0+\frac{\delta}{8}\right\}.
$$
It then follows
from (\ref{C-rho}), the continuity of $f$
\cob
and the fact $y_0+\delta/8<y^*$ that $D_k \subseteq A_k, B' \subseteq B_k$ for all large enough $k$ and $B'$ has a non-empty interior.
Observing that $D^c_k$ converges
\cob
to the finite set consisting of all stationary points of $f$, we deduce that $D_k\cap B'\neq \emptyset$ and therefore $A_k\cap B_k\neq \emptyset$ for sufficiently large $k$ and therefore establish the claim. Finally, it immediately follows from this claim and the observation that $(a)$ trivially holds
\cob
for $k_0$ sufficiently large that there exists $k_0$ such that $(a)$ and $(b)$ are both satisfied.
\end{proof}

Recalling that $f$
\cob
and each $f_k$ are assumed to have finitely many stationary points in $\Theta^\circ$, we now present our second algorithm.
\begin{algorithm} \label{algo3}  (The second modified gradient ascent algorithm)

\textbf{Step $0$.}
Choose $k_0$, $y_0$ and $0<b<1$
\cob
such that Lemma \ref{initial2} is satisfied. Set $k=0$, $g_0=f_{k_0}$ and choose $\alpha \in (0, 0.5)$, $\beta \in (0, 1)$, $\theta_0 \in A_{k_0}\cap B_{k_0}$ where $A_{k_0}$ and $B_{k_0}$ are defined as in Lemma \ref{initial2}.

\textbf{Step $1$.} Increase $k$ by $1$. Set $t=1$ and $g_k=f_{k+k_0}$.

\textbf{Step $2$.} Set
$$
\tau=\theta_{k-1}+t\nabla g_{k-1} (\theta_{k-1}).
$$
If $\tau \not \in \Theta^\circ$ or
\cob
$$
||\nabla g_k (\tau)||_2< \displaystyle\frac{2N \rho^{k/3}}{1-b}
$$
or
$$
g_k(\tau) < g_k(\theta_{k-1}) + \alpha t ||\nabla g_{k-1}(\theta_{k-1})||_2^2,
$$
set $t=\beta t$ and go to Step $2$, otherwise set $\theta_k=\tau$ and go to Step $1$.
\end{algorithm}
\medskip
\begin{rem}
The constants in Step 0 are chosen to ensure the convergence of the algorithm. And the existence of $\theta_0$ follows from Lemma \ref{initial2} $(b)$.
\end{rem}

\begin{rem}
In Step 2, for any feasible $k$, one of the necessary conditions for updating the value of $\theta_k$ is
$$
\|\nabla g_k (\tau)\|_2 \geq \displaystyle\frac{2N \rho^{k/3}}{1-b}.
$$
This is a key condition imposed to make sure that $\|\nabla g_k (\tau)\|$ is not too small and thereby the algorithm
will not prematurely converge to a non-stationary point.
\cob
\end{rem}

\subsection{Convergence Analysis}

To conduct the convergence analysis of Algorithm~\ref{algo3},
we need to reformulate the algorithm via possible relabelling
of the functions $\{g_k\}_{k=0}^\infty$ and iterates $\{\theta_k\}_{k=0}^\infty$ similarly as in Section~\ref{convergence}.
For ease of presentation only, we assume in the reminder of this section that such a relabelling is not needed
and thereby $k$ actually records the number of times that Step $2$ has been executed.

The following theorem asserts the convergence of Algorithm~\ref{algo3} under some regularity conditions.
\begin{thm} \label{exponential-convergence-3}
Under the same assumptions as in Lemma \ref{initial2},
$$
\lim_{k\rightarrow \infty}g_k(\theta_k) \mbox{ exists and } \|\nabla g_k(\theta_k)\|_2\rightarrow 0,
$$
where $g_k(\theta_k)$ is defined in Algorithm \ref{algo3}.
\end{thm}
\medskip

\begin{proof}
Similarly as in Section \ref{convergence}, define
{\small\begin{align*}
T_1(k)&\triangleq \inf\{p\in \mathbb{Z}: \theta_{k-1}+\beta^p \nabla g_{k-1} (\theta_{k-1})\in \Theta^\circ\},\\
\hat{T}(k)&\triangleq \inf\left\{q\in \mathbb{Z}: q\geq T_1(k),\mbox{ } ||\nabla g_k(\theta_{k-1}+\beta^q\nabla g_{k-1}(\theta_{k-1}))||_2 \geq \frac{2N \rho^{(k+k_0)/3}}{1-b}\right\},\\
T(k)&\triangleq \inf\{r\in \mathbb{Z}: r\geq \hat{T}(k),\mbox{ } g_k(\theta_{k-1}+\beta^r \nabla g_{k-1}(\theta_{k-1}))\geq g_k(\theta_{k-1})+\alpha\beta^r ||\nabla g_{k-1} (\theta_{k-1})||_2^2\},
\end{align*}}
and
$$
T_2(k):=\hat{T}(k)-T_1(k), \quad T_3(k):= T(k)-\hat{T}(k).
$$In other words, for each $k$, $T_1(k)$ can be regarded as
the number of times
that Step 2 of Algorithm \ref{algo3} has been executed before the condition $\tau\in \Theta^\circ$ is met; $T_2(k)$ can be regarded as
the number of additional times
\cob
that Step 2 of Algorithm \ref{algo3} has been executed before the condition
$$ ||\nabla g_k(\theta_{k-1}+\beta^q\nabla g_{k-1}(\theta_{k-1}))||_2 \geq \frac{2N \rho^{(k+k_0)/3}}{1-b}$$
is also met
\cob
while $T_3(k)$ can be regarded as
the number of additional times
\cob
that Step 2 Algorithm \ref{algo3} has been executed before the Armijo condition
$$
g_k(\theta_{k-1}+\beta^r \nabla g_{k-1}(\theta_{k-1}))\geq g_k(\theta_{k-1})+\alpha\beta^r ||\nabla g_{k-1} (\theta_{k-1})||_2^2
$$
is also met.
\cob
The well-definedness of $\hat{T}(k)$ is based on the fact that if {\small$\theta_{k-1}+\beta^p \nabla g_{k-1} (\theta_{k-1})\in \Theta^\circ$} for some non-negative integer $p$, then the same inequality also holds for any integer $p'>p$; and the well-definedness of ${T}(k)$ will be postponed to Step 2 of the proof detailed below.

The remainder of the proof consists of $5$ steps, with the first three ones devoted to establish the uniform boundedness of 
$T_1(k)$, $T_2(k)$ and $T_3(k)$ and thus that of $T(k)$.
\cob



\textbf{Step 1: Uniform boundedness of $T_2(k)$.}  As in the proof of Theorem \ref{exponential-convergence-1}, it can be readily verified that $T_1(k)< \infty$ for all $k\geq 0$. Hence, when considering $T_2(k)$, we assume that $\tau=\theta_{k-1}+\beta^q \nabla g_{k-1}(\theta_{k-1})$ is already in $\Theta^\circ$.

In order to prove the uniform boundedness of $T_2(k)$, we proceed by way of induction. First of all, by the definition of $g_k$ and the choice of $\theta_0$, we have {\small$
||\nabla g_0(\theta_0)||_2\geq 2N \rho^{k_0/{3}}/({1-b}).$} Now, 
assuming that for some $k=1, 2, \dots,$
\cob
\begin{equation}\label{grad for k-1}
||\nabla g_{k-1} (\theta_{k-1})||_2\geq \frac{2N \rho^{(k_0+k-1)/{3}}}{1-b},
\end{equation}
we will derive a sufficient condition
\cob
on $\beta^q$ such that $
||\nabla g_{k}(\tau)||_2\geq 2N \rho^{(k_0+k)/{3}}/({1-b})$, where we recall that $\tau$ is defined as
\begin{equation} \label{tau-k}
\tau=\theta_{k-1}+\beta^q \nabla g_{k-1}(\theta_{k-1}).
\end{equation}
To this end, we first note that by the Taylor series expansion, there exist $\xi$ and $\hat{\xi}$ in $\Theta^\circ$ such that
$$
g_k(\tau)-g_k(\theta_{k-1})=\nabla g_k (\tau)^T (\tau-\theta_{k-1})-(\theta_{k-1}-\tau)^T \frac{\nabla^2 g_k(\xi)}{2} (\theta_{k-1}-\tau)
$$
and
$$g_k(\tau)-g_k(\theta_{k-1})=\nabla g_k (\theta_{k-1})^T (\tau-\theta_{k-1})+(\tau-\theta_{k-1})^T \frac{\nabla^2 g_k(\hat{\xi})}{2} (\tau-\theta_{k-1}),
$$
which immediately imply that
\begin{align}
& \nabla g_k (\tau)^T (\tau-\theta_{k-1})-(\theta_{k-1}-\tau)^T \frac{\nabla^2 g_k(\xi)}{2} (\theta_{k-1}-\tau) \notag \\
&= \nabla g_k (\theta_{k-1})^T (\tau-\theta_{k-1})+(\tau-\theta_{k-1})^T \frac{\nabla^2 g_k(\hat{\xi})}{2} (\tau-\theta_{k-1}) \label{tay}.
\end{align}
Noting that $||\nabla^2 g_k(\xi)||_2\leq M$ for all $\xi\in \Theta^\circ$ and
\begin{equation}  \label{g-minus}
||\nabla g_k(\theta)-\nabla g_{k-1}(\theta)||_2=||\nabla f_{k+k_0}(\theta)-\nabla f_{k+k_0-1}(\theta)||_2\leq N \rho^{k+k_0}
\end{equation}
for all $\theta \in \Theta^\circ$, we deduce from (\ref{tay}) that
{\small\begin{align}
 ||\nabla g_k(\tau)||_2 \cdot ||\tau-\theta_{k-1}||_2 &\geq \nabla g_k(\theta_{k-1})^T (\tau-\theta_{k-1}) -M ||\tau-\theta_{k-1}||_2^2 \notag \\
& \geq \nabla g_{k-1} (\theta_{k-1})^T (\tau-\theta_{k-1}) -N\rho^{k+k_0} ||\tau-\theta_{k-1}||_2 -M ||\tau-\theta_{k-1}||_2^2 \label{cond 2}.
\end{align}}
\!\!Clearly, it follows from (\ref{tau-k}) that
the vectors
\cob
$\nabla g_{k-1} (\theta_{k-1})$ and $\tau-\theta_{k-1}$ have the same direction,
which means that (\ref{cond 2}) can be rewritten as
{\small\begin{align*}
||\nabla g_k(\tau)||_2 \cdot ||\tau-\theta_{k-1}||_2
\geq  ||\nabla g_{k-1}(\theta_{k-1})||_2 \cdot ||\tau-\theta_{k-1}||_2 -N\rho^{k+k_0} ||\tau-\theta_{k-1}||_2 -M ||\tau-\theta_{k-1}||_2^2.
\end{align*}}
\!\!Simplifying this inequality, we have
\begin{align} \label{grad relation}
||\nabla g_k(\tau)||_2&\geq  (1-M\beta^q) ||\nabla g_{k-1}(\theta_{k-1})||_2-N\rho^{k+k_0} \notag \\
&\geq (1-M\beta^q) \frac{2N \rho^{(k_0+k-1)/{3}}}{1-b}-N\rho^{k+k_0}.
\end{align}
Now, using the fact $1-\rho^{1/3}-\rho^{2k_0/3}>0$ (see Lemma \ref{initial2} $(a)$), (\ref{grad for k-1}) and (\ref{grad relation}), we conclude that the condition
\begin{equation} \label{beta q}
\beta^q \leq \frac{1-\rho^{1/3}-\rho^{2k_0/3}}{M}
\end{equation}
is sufficient for $
||\nabla g_{k}(\tau)||_2\geq 2N \rho^{(k+k_0)/3}/(1-b).$
In other words, the induction argument successfully proceeds as long as (\ref{beta q}) holds, and therefore $T_2(k)$ can be uniformly bounded as below: for all feasible $k$,
\begin{equation} \label{T_2}
T_2(k) \leq \max \left\{0, \frac{\log\left( \left(1-\rho^{1/3}-\rho^{2k_0/3}\right)\Big/M \right)}{\log \beta}  + 1\right\}.
\end{equation}
\cob

\textbf{Step 2: Uniform boundedness of $T_3(k)$.}
First note that from (\ref{beta q}),
\cob
if the inequality
$$
||\nabla g_k(\theta_{k-1}+\beta^q\nabla g_{k-1}(\theta_{k-1}))||_2 \geq \frac{2N \rho^{(k+k_0)/3}}{1-b}
$$
holds for some non-negative integer $q$, then it remains true for all integers $q'>q$. This observation justifies the well-definedness of $T_3(k)$. Moreover, due to the boundedness of $T_2(k)$ for each $k$
(in fact, it is uniformly bounded),
\cob
we can assume without loss of generality that $
||\nabla g_k(\tau)||_2 \geq 2N \rho^{(k+k_0)/3}/(1-b)$
is already satisfied before we proceed to establish the uniform boundedness of $T_3(k)$.

By the Taylor series expansion formula and (\ref{g-minus}), we have
\begin{align}
g_k(\tau) &\geq g_k(\theta_{k-1})+ \beta^r \nabla g_k (\theta_{k-1})^T\nabla g_{k-1}(\theta_{k-1})-\frac{M \beta^{2r} }{2} ||\nabla g_{k-1} (\theta_{k-1})||_2^2 \notag \\
&\geq g_k(\theta_{k-1})+\beta^r ||\nabla g_{k-1} (\theta_{k-1})||_2^2-\frac{M \beta^{2r} }{2} ||\nabla g_{k-1} (\theta_{k-1})||_2^2 -N \rho^{k+k_0} \beta^r ||\nabla g_{k-1} (\theta_{k-1})||_2,
\notag
\end{align}
\cob
where $\tau=\theta_{k-1}+\beta^r \nabla g_{k-1} (\theta_{k-1}).$ It then follows that the condition
\begin{equation}
\beta^r \leq \frac{1}{M}-\displaystyle\frac{2N \rho^{k+k_0-1}}{M||\nabla g_{k-1} (\theta_{k-1})||_2} \notag
\end{equation}
is sufficient to ensure
\begin{equation} \label{increasing cdt 2}
g_k(\tau)\geq g_k(\theta_{k-1})+\alpha \beta^r ||\nabla g_{k-1} (\theta_{k-1})||_2^2.
\end{equation}
 Recalling that
$$
||\nabla g_{k-1}(\theta_{k-1})||_2 \geq \frac{2N \rho^{(k+k_0-1)/3}}{1-b} \geq \frac{2N \rho^{k+k_0-1}}{1-b}
$$
\cob
with $0<b<1$, we deduce that the condition $\beta^r \leq b/M$ is sufficient for (\ref{increasing cdt 2}). In other words, we have
\begin{equation} \label{T_3}
T_3(k) \leq \max\left\{0, \frac{\log b-\log M}{\log \beta}+1 \right\}.
\end{equation}

\textbf{Step 3: Uniform boundedness of $T_1(k)$ and $T(k)$.}
In this step, we will show that $T_1(k)$ is uniformly bounded over all $k$. This, together with the established fact that $T_2(k)$ and $T_3(k)$ are both uniformly bounded, immediately implies the uniform boundedness of $T(k)$ over all $k$.

From Algorithm \ref{algo3}, $
g_k(\theta_k) \geq  g_k(\theta_{k-1}) + \alpha t ||\nabla g_{k-1}(\theta_{k-1})||_2^2$
for all $k\geq 0$, where
$\theta_k = \theta_{k-1} + t \nabla g_{k-1}(\theta_{k-1})$.
\cob
Using (\ref{C-rho}), we have
\begin{align}
g_{0} (\theta_k)&\geq g_k(\theta_{k})-\frac{N\rho^{k_0+1}}{1-\rho} \notag\\
&\geq g_{k} (\theta_{k-1})-\frac{N\rho^{k_0+1}}{1-\rho}  \notag \\
&\geq g_{k-1}(\theta_{k-1})-N\rho^{k+k_0} -\frac{N\rho^{k_0+1}}{1-\rho},\notag
\end{align}
\cob
from which, by induction on $k$,
we arrive at
\cob
\begin{align} \label{inc}
g_0(\theta_k)&\geq g_{0} (\theta_{0})-\sum_{k=1}^\infty N \rho^{k+k_0}-\frac{N\rho^{k_0+1}}{1-\rho} \notag \\
&\geq g_0 (\theta_0)- \frac{2 N \rho^{k_0}}{1-\rho},
\end{align}
\cob
for all $k\geq 0$. Recalling from Lemma \ref{initial2} and Step 0 of Algorithm \ref{algo3} that $$
\theta_0\in B_{k_0} = \{x\in \Theta: f_{k_0}(x) \geq y_0\} =\{x\in\Theta: g_{0}(x) \geq y_0\},$$
we deduce from (\ref{inc}) and Lemma \ref{initial2} that for all $k\geq 0$,
$$\theta_k \in \left\{x: g_0(x)\geq y_0-\frac{2 N \rho^{k_0}}{1-\rho} \right\}\subseteq C_{k_0} \subseteq \Theta^\circ \quad \mbox{and} \quad dist(C_{k_0}, \partial\Theta)>0,$$
where $C_{k_0}$ is defined in Lemma \ref{initial2} $(b)$.
Hence, for any non-negative integer $p$ such that
$p\geq \displaystyle \log (dist(C_{k_0}, \partial\Theta)/M) /\log \beta$,
\cob
we have $\theta_{k-1}+ \beta^p \nabla g_k(\theta_{k-1}) \in \Theta^\circ,$
establishing the following uniform bound
\begin{equation} \label{T_11}
T_1(k)\leq \max\left\{0, \frac{\log (dist(C_{k_0}, \Theta^c)/M)}{\log \beta}+1 \right\}.
\end{equation}
\cob

Finally, it is clear
from (\ref{T_2}), (\ref{T_3}), (\ref{T_11})
\cob
and the definition of $T(k)$ that there exists a non-negative integer $B$
such that, for all $k$,
\cob
\begin{equation} \label{T_k}
T(k)\leq B.
\end{equation}

\textbf{Step 4: Convergence of $g_k(\theta_k)$.}

It follows from (\ref{C-rho}), (\ref{T_k}) and the fact {\small$||\nabla g_{k-1} (\theta_{k-1})||_2 \geq 2N \rho^{(k+k_0-1)/3}/(1-b)$} that
\cob
\begin{align*}
g_k(\theta_k)&\geq g_k(\theta_{k-1})+\alpha \beta^{B+1} ||\nabla g_{k-1}(\theta_{k-1})||_2^2 \notag \\
&\geq g_{k-1}(\theta_{k-1})+\alpha \beta^{B+1} ||\nabla g_{k-1}(\theta_{k-1})||_2^2 - N \rho^{k+k_0} \notag \\
&\geq g_{k-1}(\theta_{k-1}) +\frac{4 \alpha \beta^{B+1} N^2 \rho^{2(k+k_0-1)/3}}{(1-b)^2}-N\rho^{k+k_0}.
\end{align*}
Observing that if $k$ is large enough,
$$
\frac{4 \alpha \beta^{B+1} N^2 \rho^{2(k+k_0-1)/3}}{(1-b)^2}\geq N\rho^{k+k_0},
$$
we deduce that $g_{k}(\theta_k)\geq g_{k-1} (\theta_{k-1})$ for sufficiently large $k$. Noting that (\ref{C-rho}) and the definition of $g_k$
imply that there exists $C > 0$ such that
$g_k(\theta_k)\leq C$ for all $k$, we conclude that $\lim_{k \to \infty} g_k(\theta_k)$ exists.

\textbf{Step 5: {\boldmath$\|\nabla g_k(\theta_{k})\|_2\rightarrow 0$.\unboldmath}} 

Since $
g_k(\theta_k) \geq g_{k-1}(\theta_{k-1}) +\alpha \beta^{B+1} ||\nabla g_{k-1} (\theta_{k-1})||_2^2-N\rho^{k+k_0},$ we have
$$
\sum_{k=1}^{n-1} \alpha \beta^{B+1} ||\nabla g_{k-1} (\theta_{k-1})||_2^2\leq g_n(\theta_n)-g_0(\theta_0)+\sum_{k=1}^{n-1} N\rho^{k+k_0},
$$
which, together with the uniform boundedness of $\{g_k(\theta_k)\}_{k=0}^\infty$, yields
$$
\sum_{k=1}^\infty \alpha \beta^{B+1} ||\nabla g_{k-1} (\theta_{k-1})||_2^2 < \infty.
$$
Hence, $
\lim_{n\rightarrow \infty} ||\nabla g_{k-1}(\theta_{k-1})||_2 =0$. The proof of the theorem is then complete.
\end{proof}
\medskip

\subsection{A noisy channel with two states: Gilbert-Elliott Channel} \label{gil-ell}

In this section, we consider a Gilbert-Elliott channel with a first-order Markovian input under the $(1,\infty)$-RLL constraint. To be more specific, let $\oplus$ denote binary addition and $\{S_n\}_{n=0}^\infty$ be the state process which is a binary stationary Markov chain with the transition probability matrix
$$
\left[\begin{array}{cc}
0.7 & 0.3 \\
0.3 & 0.7
\end{array}\right].
$$
We focus on the Gilbert-Elliott channel characterized by the input-output equation
\begin{equation} \label{G-E}
Y_n = X_n \oplus E_n,
\end{equation}
where $\{X_n\}_{n=1}^\infty$ is a binary first-order stationary Markov chain independent of $\{S_n\}_{n=1}^\infty$ with the transition probability matrix
$$
\left[\begin{array}{cc}
1-\theta& \theta\\
1&0
\end{array}\right],
$$
and $\{E_n\}_{n=1}^\infty$ is the noise process given by
$$
E_n=
\begin{cases}
0, \quad \mbox{with probability } 0.99, \\
1, \quad \mbox{with probability } 0.01,
\end{cases}
$$
when $S_{n-1}=0$ and
$$
E_n=
\begin{cases}
0, \quad \mbox{with probability } 0.9, \\
1, \quad \mbox{with probability } 0.1,
\end{cases}
$$
when $S_{n-1}=1$.
In other words, at time $n$, if the channel state takes the value $0$, the channel is a binary symmetric channel (BSC) with crossover probability $0.01$, and if the channel state takes the value $1$, it is a BSC with crossover probability $0.1$. 
\cob
It is worth noting that $E_n$ and $E_{n-1}$ are not statistically independent for this channel.

It can be readily checked that the aforementioned channel is a finite-state channel characterized by
$$
p(y_n, s_n|x_n,s_{n-1})=p(y_n|x_n,s_{n-1})p(s_n|s_{n-1})
$$
and the mutual information rate can be computed as
$$
I(X(\theta);Y(\theta))=\lim_{k \rightarrow \infty} H(Y_k(\theta)|Y_1^{k-1}(\theta))-H(E_k(\theta)|E_1^{k-1}(\theta)).
$$
The concavity of $I(X(\theta);Y(\theta))$ with respect to $\theta$ is not known, yet it seems that Algorithm~\ref{algo3} can be applied to effectively maximize it. More specifically, setting
$$
f_k(\theta)=H(Y_k(\theta)|Y_1^{k-1}(\theta))-H(E_k(\theta)|E_1^{k-1}(\theta)),
$$
we have applied Algorithm~\ref{algo3} with the initial point $\theta_0=0.2$ and we have obtained the following simulation results, from which one can observe fast convergence of the algorithm:
\begin{center}
\begin{tabular}{c|c|c|c}
$k$ & $\theta_{k}$ & $\nabla f_k(\theta_k)$ & $f_k(\theta_k)$ \\
7 & 0.28824 & 0.360645 & 0.327527 \\
8 & 0.378401 & 0.104901 & 0.347958 \\
9 & 0.404626 & 0.0427187 & 0.349884 \\
10 & 0.415306 & 0.0186297 & 0.350211 \\
11 & 0.417635 & 0.0134652 & 0.350248 \\
12 & 0.421001 & 0.00605356 & 0.350281 \\
13 & 0.422514 & 0.00274205 & 0.350288 \\
14 & 0.4232 & 0.0012462 & 0.350289 \\
15 & 0.423511 & 0.000567221 & 0.350289 \\
16 & 0.423653 & 0.000258353 & 0.350289
\end{tabular}
\end{center}
\begin{figure}[H]
\begin{center}
\includegraphics[width=7cm]{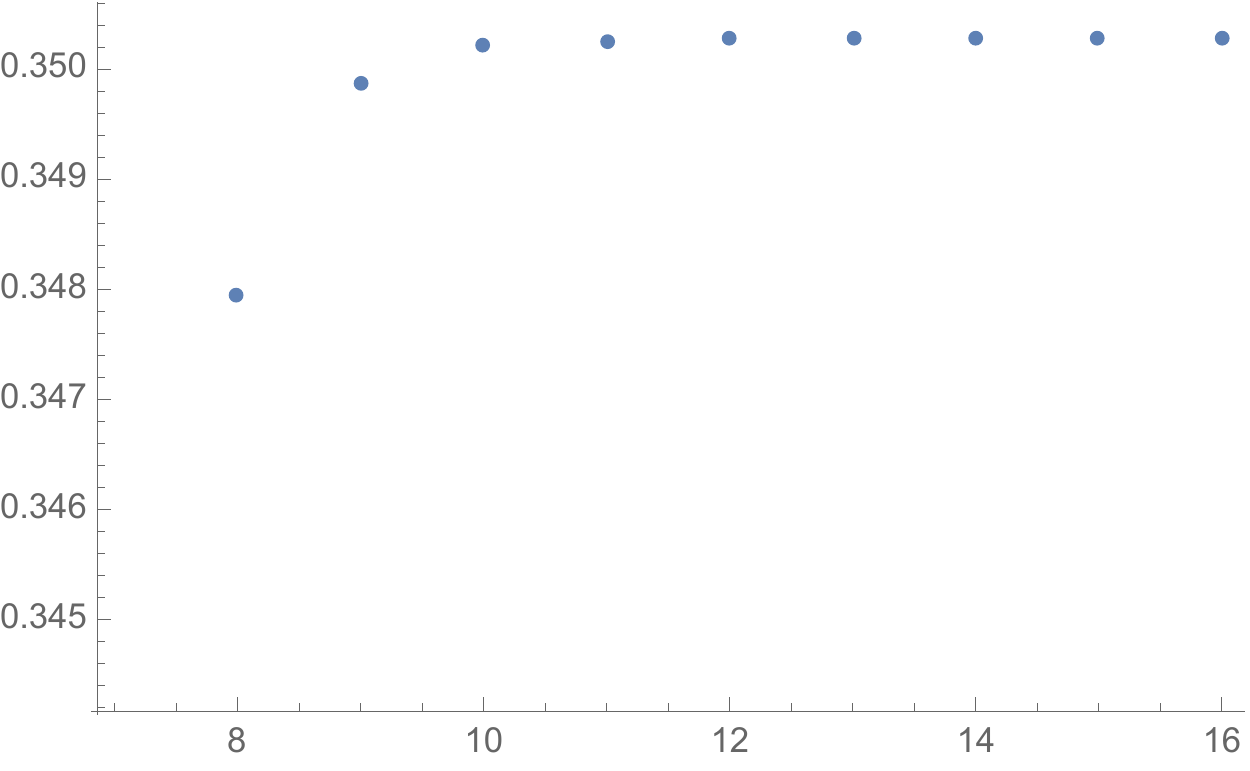}
\caption{$f_k(\theta_k)$}
\label{BECsimu}
\end{center}
\end{figure}

\bigskip \bigskip


\end{document}